\documentclass[11pt]{article} % For LaTeX2e

\usepackage{fullpage}
\usepackage{palatino}

\usepackage{amsthm,amsmath,amssymb,graphicx,epsfig,color}
\usepackage{graphics, graphicx}
\usepackage{enumerate}
\usepackage{subfigure}
\usepackage{verbatim}
\usepackage{url}
\usepackage{bm} % for bold greek symbols

\usepackage{framed}
\usepackage{amssymb,amsfonts,amsmath}
\usepackage{algorithm}
\usepackage{algorithmic}
\usepackage{amsfonts,amsmath}
\usepackage{graphicx}
\usepackage{subfigure,enumerate,amsmath,amsthm,wrapfig,array,color}
\usepackage{graphics}
\usepackage{multirow}

\title{Efficient Dimensionality Reduction for Canonical Correlation Analysis\footnote{An extended abstract of this work will appear in the 2013 International Conference of Machine Leanring (ICML).}}

\author{
{\bf Haim Avron}\\
\small Business Analytics \& Mathematical Sciences\\
IBM T. J. Watson Research Center\\
\texttt{haimav@us.ibm.com} \\
\and
{\bf Christos Boutsidis} \\
\small Business Analytics \& Mathematical Sciences\\
IBM T. J. Watson Research Center\\
\texttt{cboutsi@us.ibm.com} \\
\and
{\bf Sivan Toledo } \\
{\small Blavatnik School of Computer Science} \\
Tel-Aviv University \\
\texttt{stoledo@tau.ac.il} \\
\and
{\bf Anastasios Zouzias }\\
\small Mathematical \& Computational Sciences\\
IBM Z\"{u}rich Research Lab\\
\texttt{azo@zurich.ibm.com} \\
}

\long\def\symbolfootnote[#1]#2{\begingroup%
\def\thefootnote{\fnsymbol{footnote}}\footnote[#1]{#2}\endgroup}

\newcommand{\TNorm }[1]{\mbox{}\|#1\|_2  }
\newcommand{\TNormS}[1]{\mbox{}\|#1\|_2^2}

\newtheorem{theorem}{\bf Theorem}[]
\newtheorem{lemma}[theorem]{Lemma}
\newtheorem{definition}[theorem]{Definition}

\newtheorem{corollary}[theorem]{Corollary}

\newcommand{\transp}{^{\textsc{T}}}

\newcommand{\mat}[1]{{\ensuremath{\bm{\mathrm{#1}}}}}

\newcommand{\abs }[1]{\left|#1\right|}

\newcommand{\pinv}[1]{ {#1}^{+}}
\newcommand{\tpinv}[1]{ (\pinv{#1})\transp}

\def\one{{\bm{1}}}

\def\rank{\hbox{\rm rank}}

\def\b{{\mathbf b}}
\def\e{{\mathbf e}}

\def\p{{\mathbf p}}

\def\u{{\mathbf u}}
\def\v{{\mathbf v}}

\def\matA{\mat{A}}
\def\matB{\mat{B}}
\def\matC{\mat{C}}
\def\matD{\mat{D}}
\def\matE{\mat{E}}
\def\matF{\mat{F}}
\def\calF{{\cal F}}
\def\matG{\mat{G}}
\def\matH{\mat{H}}
\def\matI{\mat{I}}
\def\matL{\mat{L}}

\def\matP{\mat{P}}
\def\matQ{\mat{Q}}
\def\matR{\mat{R}}
\def\matS{\mat{S}}

\def\matU{\mat{U}}
\def\matV{\mat{V}}
\def\matW{\mat{W}}
\def\matX{\mat{X}}
\def\matY{\mat{Y}}
\def\matZ{\mat{Z}}
\def\matSig{\mat{\Sigma}}
\def\matTh{\mat{\Theta}}
\def\matOmega{\mat{\Omega}}

\def\matPhi{\mat{\Phi}}
\def\matPsi{\mat{\Psi}}

\newcommand{\infnorm}[1]{\ensuremath{\|#1\|_\infty}}

\def\w{{\mathbf{w}}}

\DeclareMathSymbol{\Prob}{\mathbin}{AMSb}{"50}
\newcommand\remove[1]{}
\newcommand\ignore[1]{}

\def\nnz{{ \rm nnz }}

\def\math#1{$#1$}

\def\frac#1#2{{#1\over #2}}

\DeclareMathSymbol{\R}{\mathbin}{AMSb}{"52}
\def\p{{\mathbf p}}
\def\x{{\mathbf x}}
\def\y{{\mathbf y}}
\def\z{{\mathbf z}}

\def\a{{\mathbf a}}
\def\b{{\mathbf b}}

\def\norm#1{{\|#1\|}}

\def\ceil#1{{\left\lceil\,#1\,\right\rceil}}

\def\dotfil{\leaders\hbox to 1.5mm{.}\hfill}

\def\eps{\epsilon}

\newcommand{\ip}[2]{\left\langle {#1},\ {#2} \right\rangle}

\begin{document}

\maketitle

% INCLUDES the content of the paper (Tassos)

% This is the main content of the CCA paper (Tassos)
% This file is included in the paperCCA.tex source file and it is also included under the ../SISC_SIAM_format/paperCCA_SIAM.tex file (using the SIAM format)

\begin{abstract}
%An efficient algorithm for approximate Canonical Correlation Analysis (CCA) is presented.
We present a fast algorithm for approximate Canonical Correlation Analysis (CCA).
Given a pair of tall-and-thin matrices, the proposed algorithm first employs a randomized
dimensionality reduction transform to reduce the size of the input matrices,
and then applies any CCA algorithm to the new pair of matrices.
The algorithm computes an approximate CCA to the original pair of matrices with provable guarantees, while requiring asymptotically less operations than the state-of-the-art exact algorithms.
%We also demonstrate that our algorithm is faster than the standard algorithm in practice by 30-60\% 
%even on fairly small matrices.

%Canonical Correlation Analysis (CCA) is an important technique in statistics, data analysis, and data mining that has been successfully applied in several machine learning applications. However, the cost of CCA is prohibitive for large data-sets. We show how dimensionality reduction can be used to accelerate CCA. Our main result is a simple randomized algorithm that returns approximations to the canonical correlations of a matrix pair. Our algorithm first reduces the dimension of the two matrices via the Subsampled Randomized Walsh-Hadamard Transform, and then applies any CCA algorithm to the down-sampled matrices. In most application, the dimensionality of the subspaces is sufficiently smaller compared to the number of dimensions of the host space. In this setting, the new algorithm is computationally superior compared to any known (exact) method. We analyze the proposed randomized technique both theoretically and empirically.
\end{abstract}

\section{Introduction}

Canonical Correlation Analysis (CCA)~\cite{Hot36} is an important technique in statistics, data analysis, and data mining. CCA  has been successfully applied in many statistics and machine learning applications, e.g. dimensionality reduction~\cite{SCY10}, clustering~\cite{CKLS09}, learning of word embeddings~\cite{DFU11}, sentiment classification~\cite{DRFU12}, discriminant learning~\cite{SFGT12}, and object recognition~\cite{KKC07}.
 %and activity recognition from video~\cite{LAMCS11}.
In many ways CCA is analogous to Principal Component Analysis (PCA), but instead of analyzing a single data-set (in matrix form), the goal of CCA is to analyze the relation between a pair of data-sets (each in matrix form). From a statistical point of view, PCA extracts the maximum covariance directions between elements in a single matrix, whereas CCA finds the direction of maximal correlation between a pair of matrices.
% I don't understand what it means to "adopt a point-of-view". You can "have a point-of-view" or "take a point-of-view".
From a linear algebraic point of view, CCA measures the similarities between two subspaces (those spanned by the columns of each of the two matrices analyzed). From a geometric point of view, CCA computes the cosine of the \emph{principle} angles between the two subspaces.
% I moved the Geo-PoV here since we talk about PoV here and not later.

There are different ways to define the canonical correlations of a pair of matrices, and all these methods are equivalent~\cite{GZ95}.
%From an application point of view, statistically-oriented definitions are often the most appropriate, but
The linear algebraic formulation of Golub and Zha~\cite{GZ95}, which we present shortly,
serves our algorithmic point of view best.
\begin{definition}\label{def}
Let $\matA \in \R^{m \times n}$ and $\matB \in \R^{m \times \ell}$ , and assume that $p = \rank(\matA) \geq \rank(\matB) = q$.
The {\em canonical correlations}
$$\sigma_1\left( \matA, \matB \right) \ge \sigma_2\left( \matA, \matB \right) \ge \cdots \ge \sigma_q\left( \matA, \matB \right)$$
of the matrix pair $(\matA, \matB)$ are defined recursively by the following formula:
\begin{eqnarray*}
\sigma_i\left(\matA, \matB \right) = \max_{ \x \in {\cal A}_i, \y \in {\cal B}_i }  \sigma \left( \matA \x, \matB \y \right) = : \sigma\left( \matA \x_i, \matB \y_i \right) ,
\quad i=1,\ldots ,q
\end{eqnarray*}
where
\begin{itemize}
	
\item $ \sigma\left(\u, \v \right)  = | \u\transp \v | / \left( \TNorm{\u} \TNorm{\v} \right)$,

\item $ {\cal A}_i = \{ \x : \matA \x \neq \bf{0}, \matA \x \perp \{ \matA \x_1,\ldots,\matA \x_{i-1} \} \} $,

\item $ {\cal B}_i = \{ \y : \matB \y \neq \bf{0}, \matB \y \perp \{ \matB \y_1,\ldots,\matB \y_{i-1} \} \} $.
\end{itemize}
The unit vectors
\begin{eqnarray*} \matA \x_1 / \TNorm{\matA \x_1}, \dots, \matA \x_q / \TNorm{\matA \x_q},
\quad and
\quad  \matB \y_1 / \TNorm{\matB \y_1}, \dots, \matB \y_q / \TNorm{\matB \y_q},
\end{eqnarray*}
 are called the {\em canonical} or {\em principal vectors}.
The vectors
\begin{eqnarray*}
\x_1 / \TNorm{\matA \x_1}, \dots, \x_q / \TNorm{\matA \x_q},
\quad and
\quad \y_1 / \TNorm{\matB \y_1}, \dots, \y_q / \TNorm{\matB \y_q},
\end{eqnarray*}
are called {\em canonical weights} (or {\em projection vectors}). Note that the canonical weights and the canonical vectors are \emph{not} uniquely defined.
\end{definition}

\subsection{Main Result}

The main contribution of this article (see Theorem~\ref{thm:alg}) is a fast algorithm to compute an approximate CCA. The algorithm computes an additive-error approximation to \emph{all} the canonical correlations. It also computes a set of approximate canonical weights with provable guarantees. We show that the proposed algorithm is asymptotically faster compared to the standard method of Bj{\"o}rck and Golub~\cite{BG73}. To the best of our knowledge, this is the first sub-cubic time algorithm for approximate CCA that has provable guarantees.

The proposed algorithm is based on \emph{dimensionality reduction}: given a pair of matrices $(\matA, \matB)$, we transform the pair to a new pair $(\hat{\matA}, \hat{\matB})$ that has much fewer rows, and then compute the canonical correlations of the new pair exactly, alongside a set of canonical weights, e.g. using the Bj{\"o}rck and Golub algorithm.
We prove that with high probability the canonical correlations of $(\hat{\matA}, \hat{\matB})$ are close to the canonical correlations of $(\matA, \matB)$, and that any set of canonical weights of $(\hat{\matA}, \hat{\matB})$ can be used to construct a set of approximately orthogonal canonical vectors of $(\hat{\matA}, \hat{\matB})$.
%Now, any CCA algorithm can be applied on $(\hat{\matA}, \hat{\matB})$; in our analysis we assume that the Bj{\"o}rck and Golub algorithm is used.
The transformation of $(\matA, \matB)$ into $(\hat{\matA}, \hat{\matB})$ is done in two steps. First, we apply the \emph{Randomized Walsh-Hadamard Transform (RHT)} to both $\matA$ and $\matB$. This is a unitary transformation, so the canonical correlations are preserved exactly. On the other hand, we show that with high probability, the transformed matrices have their ``information'' equally spread among all the input rows, so now the transformed matrices are amenable to uniform sampling. In the second step, we uniformly sample (without replacement) a sufficiently large set of rows and rescale them to form $(\hat{\matA}, \hat{\matB})$. The combination of RHT and uniform sampling is often called \emph{Subsampled Randomized Walsh-Hadamard Transform (SRHT)} in the literature~\cite{Tro11}. Note that
other variants of dimensionality reduction~\cite{Sar06} might be appropriate as well, but for concreteness we focus on the SRHT (see also Section~\ref{sec:sparse}).
% witout "as well" appropriate will imply "better", which we do not want to say (and don't think so)

Our dimensionality reduction scheme is particularly effective when the matrices are tall-and-thin, that is they have much more rows than columns. Targeting such matrices is natural: in typical CCA applications, columns typically correspond to features or labels and rows correspond to samples or training data. By computing the CCA on as many instances as possible (as much training data as possible), we get the most reliable estimates of application-relevant quantities.
However in current algorithms adding instances (rows) is expensive, e.g. in Bj{\"o}rck and Golub algorithm we pay $O(n^2+\ell^2)$ for each row. Our algorithm allows practitioners to run CCA on huge data sets because we reduce the cost of an extra row to almost $O(n+\ell)$.

We also discuss a variant of our dimensionality reduction scheme that is more suitable for sparse matrices (Section~\ref{sec:sparse}), and show that it is not possible to replace the additive error guarantees in our analysis with relative error guarantees (Section~\ref{sec:error:lowerbound}). Finally, we demonstrate that our algorithm is faster than the standard algorithm in practice by 30-60\% even on fairly small matrices  (Section~\ref{sec:experiments}).%, while still attaining a good approximate CCA (Section~\ref{sec:experiments}.

\subsection{Related Work}

Dimensionality reduction has been the driving force behind many recent algorithms for accelerating key machine learning and linear algebraic tasks. A representative example is linear regression, i.e., solve the least squares problem $\min_{\x}\TNorm{\matA \x - \b}$, where $\b \in \R^{m}$. If $m \gg n$, then one can use the SRHT to reduce the dimension of $\matA$ and $\b$, to form $\hat{\matA}$ and $\hat{\b}$, and then solve the small problem $\min_{\x}\TNorm{\hat{\matA} \x - \hat{\b}}$. This process will return an approximate solution to the original problem~\cite{Sar06,BD09,DMMS11}. Alternatively, one can observe that $\matA\transp \matA$ and $\hat{\matA}\transp \hat{\matA}$ are spectrally close, so $\hat{\matA}$ is an effective preconditioner for $\matA$~\cite{RT08,AMT10}. Other problems that can be accelerated using dimensionality reduction include:
(i) approximate PCA (via low-rank matrix approximation)~\cite{HMT};
(ii) matrix multiplication~\cite{Sar06};
(iii) K-means clustering~\cite{BZD10};
(iv) approximation of matrix coherence and statistical leverage~\cite{DMMW12}; to name only a few.

%{\bf Christos: I believe we need such an argument to avoid comments like ''where is the novelty here?''}
Our approach uses similar techniques as the algorithms mentioned above. For example, Lemma~\ref{lemma:sampling-ortho}
plays a central role in these algorithms as well. However, our analysis requires the use of advanced ideas from matrix perturbation theory and it leads to two new technical lemmas that might be of independent interest:
Lemmas~\ref{lem:pert4} and~\ref{lem:pert5} provide bounds for the singular values of the product of two \emph{different} sampled orthonormal matrices.
Previous work only provides bounds for products of the \emph{same} matrix (Lemma~\ref{lemma:sampling-ortho};
see also~\cite[Corollary 11]{Sar06})
%We believe that these two results will find other uses in the analysis of dimensionality reduction techniques for matrix problems.

Dimensionality reduction techniques for accelerating CCA have been suggested or used in the past.
One common technique is to simply use less samples by uniformly sampling the rows.
Although this technique might work reasonably well in many instances, it may fail for others
unless all rows are sampled. In fact, Theorem~\ref{thm1}
analyzes uniform sampling, and establishes bounds on the required sample size.

Sun et al. suggest a two-stage approach which involves first solving a least-squares problem, and then using the solution to reduce the problem size~\cite{SCY10}. However, their technique involves explicitly factoring one of the two matrices, which takes cubic time. Therefore, their method is especially effective when one of the two matrices has significantly less columns than the other. When the two matrices have about the same number of columns, there is no asymptotic performance gain. In contrast, our method is sub-cubic in any case.

Finally, it is worth noting that CCA itself has been used for dimensionality reduction~\cite{SJY08, CKLS09, SCY10}.
This is not the focus of this article; we suggest a dimensionality reduction technique to accelerate CCA.

\section{Preliminaries}
%

%
%\subsection{Preliminaries}
%
We use $i:j$ to denote the set $\{i,\dots,j\}$, and $[n]=1:n$. We use \math{\matA,\matB,\ldots} to denote matrices and \math{\a,\b,\ldots} to denote column vectors.
$\matI_{n}$ is the $n \times n$ identity matrix;  $\bm{0}_{m \times n}$ is the $m \times n$ matrix of zeros.
%We use the Frobenius and the spectral norm of a matrix: $ \FNorm{\matA} = \sqrt{\sum_{i,j} \matA_{ij}^2}$ and
%$\TNorm{\matA} = \max_{\x:\TNorm{\x}=1}\left(\x\transp\matA\transp \matA \x\right)
%= \max_{\x,\y:\TNorm{\x}=\TNorm{\y}=1}\left(\x\transp\matA\transp \matA \y\right)$, respectively.
We denote the number of non-zero elements in $\matA$ by $\nnz(\matA)$. We denote by $\mathcal{R}(\cdot)$ the column space of its argument matrix. We denote by $[\matA ; \matB]$ the matrix
obtained by concatenating the columns of $\matB$ next to the columns of $\matA$. Given a subset of indices $T \subseteq [m]$, the corresponding sampling matrix $\matS$ is the $|T|\times m$ matrix obtained by discarding from $\matI_{m}$ the rows whose index is not in $T$. Note that $\matS \matA$ is the matrix obtained by keeping only the rows in $\matA$ whose index {\em appears} in $T$.
A symmetric matrix $\matA$ is positive semi-definite (PSD), denoted by $0 \preceq \matA$, if $\x\transp \matA \x \geq 0$ for every vector $\x$.
%Equivalently, all eigenvalues of $\matA$ are non-negative.
For any two symmetric matrices $\matX$ and $\matY$ of the same size, $\matX \preceq \matY$ denotes
%$0 \preceq \matY - \matX$, i.e. $\matY - \matX$ is PSD.
that $\matY - \matX$ is a PSD matrix.

%Given the $\mbox{SVD}$ of an $m\times n$ matrix $\matA$ with $m \geq n$, viz. $\matA=\matU \matSig \matV\transp$,
%the {\em reduced} $\mbox{SVD}$ factorization consists of discarding the last $m-n$ columns of $\matU$ and bottom $m-n$ rows of $\Sigma$. By removing singular triplets associated with zero singular values we get the {\em compact} (or {\em thin}) $\mbox{SVD}$.
We denote the  {\em compact} (or {\em thin}) $\mbox{SVD}$ of a matrix $\matA \in \R^{m \times n}$ of rank $p$ by $\matA = \matU_{\matA} \matSig_\matA \matV_\matA\transp$, with $\matU_{\matA} \in \R^{m \times p}$, $\matSig_\matA \in \R^{p \times p}$, and $\matV_\matA\transp \in \R^{p \times n}$. The Moore-Penrose pseudo-inverse of $\matA$ is
$\pinv{\matA} = \matV_\matA \matSig_\matA^{-1} \matU_\matA\transp \in \R^{n \times m}$. We denote the singular values of $\matA$ by $\sigma_1(\matA) \geq \sigma_2(\matA) \geq \dots \geq \sigma_{p}(\matA)$.

%We now state a few known results that we use later.
%\begin{lemma}[\cite{EI95} Theorem 3.3]\label{lem:pert1}
%Let $\matPsi \in \R^{p \times q}$ and $ \matPhi =  \matD_L \matPsi \matD_R $ with $\matD_L \in \R^{p \times p}$ and $\matD_R \in \R^{q \times q}$ being non-singular matrices.
%Let $\gamma = \max\{  \TNorm{ \matD_L \matD_L\transp - \matI_p },  \TNorm{ \matD_R\transp \matD_R - \matI_q } \} $. Then, for all $i=1,\ldots,\rank(\matPsi):$
%$  |  \sigma_i\left( \matPhi \right) -  \sigma_i\left( \matPsi \right)   |  \le \gamma \cdot  \sigma_i\left( \matPsi \right). $
%\end{lemma}
%
%\begin{lemma}[Weyl's inequality for singular values; \cite{HJ85} Corollary 7.3.8]\label{lem:pert2}
%Let $\matPhi, \matPsi \in \R^{m \times n}$. Then, for all $i=1,\ldots, \min( m,n):$
%$| \sigma_i\left(\matPhi\right)- \sigma_i\left(\matPsi\right)  |  \le \TNorm{\matPhi - \matPsi} $.
%\end{lemma}
%
%\begin{lemma}[Conjugating the PSD ordering; Observation 7.7.2 in~\cite{HJ85}]\label{lem:pert3}
%Let $\matPhi, \matPsi \in \R^{n \times n}$ are symmetric matrices. Let $\matPhi \preceq \matPsi$. Then, for every $n \times m$ matrix $\matZ:$
%$\matZ \transp \matPhi \matZ \preceq \matZ\transp \matPsi \matZ$.
%\end{lemma}

\subsection{The Bj{\"o}rck and Golub Algorithm}

There are quite a few algorithms to compute the canonical correlations~\cite{GZ95}. One of the most popular methods is due to Bj{\"o}rck and Golub~\cite{BG73}. It is based on the following observation.
\begin{theorem}[\cite{BG73}]\label{thm:bjork-golub}
Assume that the columns of $\matQ \in \R^{m \times p}$ ($m \geq p$) and $\matW \in \R^{m \times q}$ ($m \geq q$) form an
orthonormal basis for the range of $\matA$ and $\matB$ (respectively). Let $\matQ\transp \matW=\matU \matSig \matV\transp$ be its compact SVD. The diagonal elements of $\matSig$ are the canonical correlations of $(\matA, \matB)$. The canonical vectors are given by the first $q$ columns of $\matQ \matU$ (for $\matA$) and $\matW \matV$ (for $\matB$).
\end{theorem}
%The canonical correlations of the pair $(\matA,\matB)$ is a property of the subspace spanned by $\matA$ and $\matB$. So,

Theorem~\ref{thm:bjork-golub} implies that once we have a pair of matrices $\matQ$ and $\matW$ with orthonormal columns whose column space spans the same column space of $\matA$ and $\matB$, respectively, then all we need is to compute the singular value decomposition of $\matQ\transp \matW$.  Bj{\"o}rck and Golub suggest the use of QR decompositions, but  $\matU_\matA$ and $\matU_\matB$ will serve as well. Both options require $ O \left(m\left(n^2 + \ell^2\right) \right)$ time.
%; we use the latter approach here.
%
\begin{corollary}\label{cor:bjork-golub}
Frame Definition~\ref{def}. Let $\matU\transp_\matA \matU_\matB=\matU \matSig \matV\transp$ be the compact SVD of $\matU\transp_\matA \matU_\matB$. Then, for $i\in[q]$:
$
\sigma_i(\matA, \matB) = \matSig_{ii}$. The canonical weights are given by the columns of $\matV_\matA \matSig^{-1}_{\matA} \matU$ (for $\matA$) and $\matV_\matB \matSig^{-1}_{\matB} \matV$ (for $\matB$).
\end{corollary}
\subsection{Matrix Coherence and Sampling from an Orthonormal Matrix}
Matrix coherence is a fundamental concept in the analysis of matrix sampling algorithms (e.g.~\cite{TR10,IW12}).
There a quite a few similar but different ways to define the coherence. In
this article we use the following definition. Given a matrix $\matA$ with $m$ rows,
%the squared row norms of $\matU_\matA$ are the {\em statistical leverages} of the rows of $\matA$ which represent the relative importance of the corresponding row in $\matA$, e.g. in linear regression.
the \emph{coherence} of $\matA$ is defined as
$$\mu(\matA)=\max_{i\in{[m]}}\TNormS{ \e_i^\top \matU_{\matA}  },$$
where $\e_i$ is the $i$-th standard basis (column) vector of $\R^m$.
%Coherence is an important quantity in the analysis of randomized matrix algorithms.
Note that the coherence of $\matA$ is a property of the column space of $\matA$, and does not depend on the actual choice of $\matA$. Therefore, if $\mathcal{R}(\matA) = \mathcal{R}(\matB)$ then $\mu(\matA) = \mu(\matB)$. Furthermore, it is easy to verify that if $\mathcal{R}(\matA) \subseteq \mathcal{R}(\matB)$ then $\mu(\matA) \leq \mu(\matB)$. Finally, we mention that for every matrix $\matA$ with $m$ rows:
$\rank(\matA)/m \leq \mu(\matA) \leq 1.$

We focus on tall-and-thin matrices, i.e. matrices with (much) more rows than columns.
We are interested in  dimensionality reduction techniques that (approximately) preserve
the singular values of the original matrix. The simplest idea to do dimensionality reduction
in  tall-and-thin matrices is uniform sampling of the rows of the matrix.
Coherence measures how susceptible the matrix is to uniform sampling; the following lemma shows that
not too many samples are required when the coherence is small. The bound is almost tight~\cite[Section 3.3]{Tro11}.
%High coherence implies, intuitively, that there exists rows in $\matA$ which hold valuable information regarding $\matA$. Matrices with large coherence present a challenge for algorithms based on uniform sampling of rows: when the number of rows is large, those crucial rows are not likely to be selected unless the sample size is increased accordingly. We see that in the following lemma, where coherence affects the amount of rows needed to be sampled from a matrix with orthonormal columns so that the sampled matrix remains close to orthonormal (i.e., its singular values are close to one).

\begin{lemma}[Sampling from Orthonormal Matrix, Corollary to Lemma~3.4 from~\cite{Tro11}]
\label{lemma:sampling-ortho}
Let $\matQ \in \R^{m \times d}$ have orthonormal columns. Let $0 < \epsilon < 1$ and $0 < \delta < 1$.
Let $r$ be an integer such that
\[
6 \epsilon^{-2} m  \mu(\matQ) \log (3d/\delta) \leq r \leq m \,.
\]
Let $T$ be a random subset of $[m]$ of cardinality $r$, drawn from a uniform distribution over such subsets,
and let $\matS$ be the $|T|\times m$ sampling matrix corresponding to $T$ rescaled by $\sqrt{m/r}$.
Then, with probability of at least $1-\delta$, for $i\in[d]$:
$  \sqrt{1-\epsilon} \le \sigma_i(\matS \matQ) \le  \sqrt{1+\epsilon}.$
%we have $ \sqrt{1-\epsilon} \leq \sigma_d(\matS \matQ) \leq \sigma_{d-1}(\matS \matQ)  \leq \dots \leq \sigma_1(\matS \matQ) \leq \sqrt{1 + \epsilon}\,.$
\end{lemma}
\begin{proof}
Apply Lemma 3.4 from~\cite{Tro11} with the following choice of parameters:
$\ell = \alpha M \log(k/\delta),$
$\alpha = 6/\epsilon^2,$ and
$\delta_{tropp} = \eta = \epsilon$.
Here, $\ell$, $\alpha$, $M$, $k$, $\eta$ are the parameters of  Lemma 3.4 from~\cite{Tro11};
also $\delta_{tropp}$ plays the role of $\delta$, an error parameter, of  Lemma 3.4 from~\cite{Tro11}.
$\epsilon$ and $\delta$ are from our Lemma.
%A different choice of parameters gives Theorem~3.1 in \cite{Tro11}. The choice of $\ell$ proportional to $\log(k/\delta)$ rather than proportional to $\log(k)$, as in the original statement of Lemma~3.4, is what results in a probability proportional to $\delta$ instead of $k$; this can easily be seen by tracing the modified choice of $\ell$ through the proof of Lemma~3.4.
\end{proof}

In the above lemma, $T$ is obtained by sampling coordinates from $[m]$ \emph{without} replacement. Similar results can be shown for sampling with replacement, or using Bernoulli variables~\cite{IW12}.

\subsection{Randomized Fast Unitary Transforms}\label{sec:wht}

Matrices with high coherence pose a problem for algorithms based on uniform row sampling. One way to circumvent this problem is to use a coherence-reducing transformation. It is important that this transformation will not change the solution to the problem.

One popular coherence-reducing method is applying a randomized fast unitary transform.  The crucial observation is that many problems can be safely transformed using unitary matrices. This is also true for CCA: $\sigma_i(\matQ \matA, \matQ \matB) = \sigma_i(\matA, \matB)$ if $\matQ$ is unitary (i.e., $\matQ\transp \matQ$ is equal to the identity matrix). If the unitary matrix is chosen carefully, it can reduce the coherence. However, any fixed unitary matrix will fail to reduce the coherence on some matrices.

The solution is to couple a fixed unitary transform with some randomization. More specifically, the construction is 
$ \calF=  \matF  \matD\,$,
where $\matD$ is a random diagonal matrix of size $m$ whose entries are independent random signs,
and $\matF$ is some fixed unitary matrix. An important quantity is 
the maximum squared element in $\matF$ (we denote this quantity with $\eta$): for any fixed $\matX \in \R^{n\times m}$ it can be shown that with constant probability, 
$\mu(\calF \matX) \leq O\left( \eta n  \log(m) \right)$~\cite{AMT10}. So, it is important for $\eta$ to be small. It is also necessary that $\matF$ can be applied quickly to $\matX$. FFT and FFT-like transforms have both these properties, and work well in practice due to the availability of high quality implementations. 

Another fast unitary transform that has the above two properties is the Walsh-Hadamard Transform (WHT), which is defined as follows.
Fix an integer $m = 2^h$, for $h = 1,2,3, \ldots$. The (non-normalized) $m \times m$ matrix of the Walsh-Hadamard Transform (WHT) is defined recursively as,
$$ \matH_m = \left[
\begin{array}{cc}
  \matH_{m/2} &  \matH_{m/2} \\
  \matH_{m/2} & -\matH_{m/2}
\end{array}\right],
\ \mbox{with} \
\matH_2 = \left[
\begin{array}{cc}
  +1 & +1 \\
  +1 & -1
\end{array}\right].
$$
The $m \times m$ normalized matrix of the Walsh-Hadamard transform is $\matH = m^{-\frac{1}{2}} \matH_m$.
%\end{definition}
%

The recursive nature of the WHT allows us to compute $\matH \matX$ for an $m \times n$ matrix $\matX$ in time $O(m n \log(m))$.
However, in our case we are interested in $\matS \matH \matX$ where $\matS$ is a $r$-row sampling matrix. To compute $\matS \matH \matX$ only $O(m n \log(r))$ operations
suffice~\cite[Theorem 2.1]{AL08}.

Combining the WHT with a random diagonal sign matrix is called the Randomized Walsh-Hadamard Transform (RHT) 
\begin{definition}[Randomized Walsh-Hadamard Transform (RHT)]
\label{def:rht}
Let $m = 2^h$ for some positive integer $h$. A \emph{Randomized Walsh-Hadamard Transform (RHT)} is an $m \times m$ matrix of the form $$ \matTh =  \matH  \matD\,$$
where $\matD$ is a random diagonal matrix of size $m$ whose entries are independent random signs,
and $\matH$ is a normalized Walsh-Hadamard matrix of size $m$.
\end{definition}

For concreteness, our analysis uses the RHT since it has the tightest coherence reducing bound. 
Our results generalize to other randomized fast unitary transforms, perhaps with some slightly different bounds. 

\begin{lemma} [RHT bounds Coherence, Lemma 3.3 from~\cite{Tro11}]
\label{lem:rht-reduce}
Let $\matA$ be an $m\times n$  ($m \ge n$, $m=2^h$ for some positive integer $h$) matrix, and let $\matTh$ be an RHT. Then, with probability of at least $1-\delta$,
\[\mu(\matTh \matA) \leq \frac1{m}\left( \sqrt{n} + \sqrt{8 \log(m / \delta)} \right)^2\,.\]
\end{lemma}

\section{Perturbation Bounds for Matrix Products}\label{sec:pert}
This section states three new technical lemmas which analyze the perturbation of the singular values
of the product of a pair of matrices after dimensionality reduction.
These lemmas are essential for our analysis in subsequent sections, but they might be of independent interest as well.
We first state three well known results.
%The proofs of all three lemmas appear in the Supplementary Material.
%
%
\begin{lemma}[\cite{EI95} Theorem 3.3]\label{lem:pert1}
Let $\matPsi \in \R^{p \times q}$ and $ \matPhi =  \matD_{\matL} \matPsi \matD_{\matR} $
with $\matD_{\matL} \in \R^{p \times p}$ and $\matD_{\matR} \in \R^{q \times q}$ being non-singular matrices. 
Let $$\gamma = \max\{  \TNorm{ \matD_{\matL} \matD_{\matL}\transp - \matI_p },  \TNorm{ \matD_{\matR}\transp \matD_{\matR} - \matI_q } \}.$$
Then, for all $i=1,\ldots,\rank(\matPsi):$
$  |  \sigma_i\left( \matPhi \right) -  \sigma_i\left( \matPsi \right)   |  \le \gamma \cdot  \sigma_i\left( \matPsi \right). $
\end{lemma}
\begin{lemma}[Weyl's inequality for singular values; \cite{HJ85} Corollary 7.3.8]\label{lem:pert2}
Let $\matPhi, \matPsi \in \R^{m \times n}$. Then, for all $i=1,\ldots, \min( m,n):$
$$| \sigma_i\left(\matPhi\right)- \sigma_i\left(\matPsi\right)  |  \le \TNorm{\matPhi - \matPsi}.$$
\end{lemma}
\begin{lemma}[Conjugating the PSD ordering; Observation 7.7.2 in~\cite{HJ85}]\label{lem:pert3}
Let $\matPhi, \matPsi \in \R^{n \times n}$ be symmetric matrices with $\matPhi \preceq \matPsi$. Then, for every $n \times m$ matrix $\matZ:$
$$\matZ \transp \matPhi \matZ \preceq \matZ\transp \matPsi \matZ.$$
\end{lemma}
We now present the new technical lemmas.
\begin{lemma}\label{lem:pert4}
Let $\matA \in \R^{m \times n}$ ($m \geq n$)  and $\matB \in \R^{m \times \ell}$ ($m \geq \ell$). Define $\matC := [\matA ; \matB] \in \R^{m\times(n+\ell)}$, and suppose $\matC$ has rank $\omega$, so $\matU_{\matC}\in\R^{m\times\omega}$.  Let $\matS \in \R^{r \times m}$ be any matrix such that %$\rank( \matS \matU_{\matC})=\omega$
$$ \sqrt{1-\epsilon} \leq \sigma_{\omega}\left(\matS \matU_{\matC} \right) \leq \sigma_1\left(\matS \matU_{\matC} \right)  \le \sqrt{1+\epsilon},$$
for some $0 < \epsilon < 1$ . Then, for $i=1,\dots,\min(n,\ell)$,
$$|\sigma_i \left( \matA\transp\matB \right)  - \sigma_i\left( \matA\transp \matS \transp \matS \matB \right)| \le \epsilon\cdot\TNorm{\matA}\cdot\TNorm{\matB}\,.$$
\end{lemma}
\begin{proof}
%\paragraph{Bounding $|\sigma_i \left( \matU_{\matA}\transp\matU_{\matB} \right)  - \sigma_i\left( \matU_{\matA}\transp \matS \transp \matS \matU_{\matB} \right)|$.}
Using Weyl's inequality for the singular values of arbitrary matrices (Lemma~\ref{lem:pert2}) we obtain,
\begin{eqnarray*}
	| \sigma_i \left( \matA\transp\matB  \right)   - \sigma_i\left( \matA\transp \matS \transp \matS \matB \right)|%\\
	 & \leq & \TNorm{  \matA\transp \matS \transp \matS \matB - \matA\transp\matB }\\
     &  =   & \TNorm{  \matV_{\matA}\matSig_{\matA}\left(\matU_{\matA}\transp \matS \transp \matS \matU_{\matB}  - \matU_{\matA}\transp\matU_{\matB} \right)\matSig_{\matB} \matV\transp_{\matB} } \\
     & \leq & \TNorm{  \matU_{\matA}\transp \matS\transp \matS \matU_{\matB} - \matU_{\matA}\transp \matU_{\matB} } \cdot\TNorm{\matA}\cdot\TNorm{\matB}\,.
\end{eqnarray*}
Next, we argue that $
\TNorm{  \matU_{\matA}\transp \matS\transp \matS \matU_{\matB} - \matU_{\matA}\transp \matU_{\matB} } \le \TNorm{\matU_{\matC}\transp \matS\transp \matS \matU_{\matC}   - \matI_{\omega}}$.
Indeed, we now have
\begin{eqnarray*}
	\TNorm{  \matU_{\matA}\transp \matS\transp \matS \matU_{\matB} - \matU_{\matA}\transp \matU_{\matB} } %\\
	& = & \sup_{\TNorm{\w}=1,\ \TNorm{\z}=1} | \w\transp \matU_{\matA}\transp \matS\transp \matS \matU_{\matB} \z - \w\transp \matU_{\matA}\transp \matU_{\matB}\z | \\
	& = & \sup_{\TNorm{\x}=\TNorm{\y} = 1,\ \x\in{\mathcal{R}(\matU_{\matA})},\ \y\in{\mathcal{R} (\matU_{\matB})} } | \x\transp \matS\transp \matS \y - \x\transp \y | \\
	& \leq & \sup_{\TNorm{\x}=\TNorm{\y} = 1,\ \x\in{\mathcal{R}(\matU_{\matC})},\ \y\in{\mathcal{R}(\matU_{\matB}) }} | \x\transp \matS\transp \matS \y - \x\transp \y | \\
	& \leq & \sup_{\TNorm{\x}=\TNorm{y} = 1,\ \x\in{\mathcal{R}(\matU_{\matC})},\ \y\in{\mathcal{R}(\matU_{\matC})}} | \x\transp \matS\transp \matS \y - \x\transp \y | \\	
	&   =  & \sup_{\TNorm{\w}=1,\ \TNorm{\z}=1} | \w\transp \matU_{\matC}\transp \matS\transp \matS \matU_{\matC} \z  - \w\transp\matU_{\matC}\transp \matU_{\matC} \z | \\	
	&   =  & \TNorm{\matU_{\matC}\transp \matS\transp \matS \matU_{\matC}   - \matI_{\omega}}.
\end{eqnarray*}
In the above, all the equalities follow by the definition of the spectral norm of a matrix while the two inequalities follow
because $\mathcal{R}(\matU_{\matA}) \subseteq \mathcal{R}(\matU_{\matC})$ and $\mathcal{R}(\matU_{\matB}) \subseteq \mathcal{R}(\matU_{\matC})$, respectively.

To conclude the proof, recall that we assumed that for $i \in [\omega]$:
$$1-\epsilon \le \lambda_i \left( \matU_{\matC}\transp \matS\transp \matS \matU_{\matC} \right) \le 1+\epsilon.$$
\end{proof}

\begin{lemma}\label{lem:pert5}
Let $\matA \in \R^{m \times n}$ ($m \geq n$) and $\matB \in \R^{m \times \ell}$ ($m \geq \ell$).  Let $\matS \in \R^{r \times m}$ be any matrix such that $\rank(\matS \matA) = \rank(\matA)$ and $\rank(\matS \matB)=\rank(\matB)$, and all singular values of $\matS \matU_{\matA}$ and $\matS \matU_{\matB}$ are inside $[\sqrt{1-\epsilon},\sqrt{1+\epsilon}]$ for some
$0 < \epsilon < 1/2$.
Then, for $i=1,\dots,\min(n,\ell)$,
$$|\sigma_i\left( \matU_{\matA}\transp \matS\transp \matS \matU_{\matB} \right) -   \sigma_i \left( \matU_{\matS\matA}\transp\matU_{\matS \matB} \right) |
\le 2 \epsilon \left( 1 + \epsilon \right)\,.$$
\end{lemma}
\begin{proof}
For every $i=1,\ldots,q$ we have,
\small
\begin{eqnarray*}
|\sigma_i\left( \matU_{\matA}\transp \matS\transp \matS \matU_{\matB} \right) - \sigma_i \left( \matU_{\matS\matA}\transp\matU_{\matS\matB} \right) |
 &=& |\sigma_i\left(  \matSig_{\matA}^{-1} \matV_{\matA}\transp \matA\transp \matS\transp \matS \matB \matV_{\matB} \matSig_{\matB}^{-1}   \right)\\
& - & \sigma_i \left( \matSig_{\matS\matA}^{-1} \matV_{\matS\matA}\transp \matA\transp \matS\transp \matS \matB \matV_{\matS \matB} \matSig_{\matS \matB}^{-1} \right) | \\
&\leq&  \gamma \cdot \sigma_i\left( \matSig_{\matA}^{-1} \matV_{\matA}\transp\matA\transp \matS\transp \matS \matB \matV_{\matB} \matSig_{\matB}^{-1} \right) \\
&=&  \gamma \cdot \sigma_i\left( \matU_{\matA}\transp \matS\transp \matS \matU_{\matB} \right) \\
&\le& \gamma \cdot \TNorm{ \matU_{\matA}\transp \matS\transp } \cdot  \sigma_i\left( \matS \matU_{\matB} \right) \\
&\le& \gamma \cdot  \left( 1 + \epsilon\right)
\end{eqnarray*}
\normalsize
with $$ \gamma = \max(
\TNorm{\matSig_{\matS\matA}^{-1} \matV_{\matS\matA}\transp \matV_{\matA} \matSig_{\matA}^2 \matV_{\matA}\transp \matV_{\matS\matA} \matSig_{\matS\matA}^{-1} - \matI_p}  ,
\TNorm{\matSig_{\matS\matB}^{-1} \matV_{\matS\matB}\transp \matV_{\matB} \matSig_{\matB}^2 \matV_{\matB}\transp \matV_{\matS\matB} \matSig_{\matS\matB}^{-1} - \matI_q}
)\,.$$
In the above, the first inequality follows using Lemma~\ref{lem:pert1}: set
\begin{eqnarray*}
\matPsi = \matSig_{\matA}^{-1} \matV_{\matA}\transp \matA\transp \matS\transp \matS \matB \matV_{\matB} \matSig_{\matB}^{-1},
\end{eqnarray*}
\begin{eqnarray*}
\matD_L := \matSig^{-1}_{\matS \matA}\matV_{\matS\matA}\transp \matV_{\matA} \matSig_{\matA}
\quad \text{and}
\quad  \matD_R:=\matSig_{\matB}\matV_{\matB}\transp \matV_{\matS\matB} \matSig_{\matS\matB}^{-1}.
\end{eqnarray*}
Moreover,
\begin{eqnarray*}
\matD_L\Psi \matD_R &=&   \left( \matSig^{-1}_{\matS \matA}\matV_{\matS\matA}\transp \matV_{\matA} \matSig_{\matA}\right) 
\left( \matSig_{\matA}^{-1} \matV_{\matA}\transp \matA\transp \matS\transp \matS \matB \matV_{\matB} \matSig_{\matB}^{-1} \right) 
\left(\matSig_{\matB}\matV_{\matB}\transp \matV_{\matS\matB} \matSig_{\matS\matB}^{-1} \right) \\  
&=&\matSig^{-1}_{\matS \matA}\matV_{\matS\matA}\transp \matV_{\matA}  \matV_{\matA}\transp \matA\transp \matS\transp \matS \matB \matV_{\matB}
\matV_{\matB}\transp \matV_{\matS\matB} \matSig_{\matS\matB}^{-1}  \\  
&=& \matSig_{\matS\matA}^{-1} \matV_{\matS\matA}\transp \matA\transp \matS\transp \matS \matB \matV_{\matS \matB} \matSig_{\matS \matB}^{-1},
\end{eqnarray*}
since $\matA = \matA \matV_{\matA}\matV_{\matA}\transp,$ and
$\matB = \matB \matV_{\matB}\matV_{\matB}\transp.$
To apply Lemma~\ref{lem:pert1} we need to show that $\matD_L$ and $\matD_R$ are non-singular. We will prove that $\matD_L$ is non-singular 
(the same argument applies to $\matD_R$). $\matD_L$ is non-singular if and only if $\matV_{\matS\matA}\transp \matV_{\matA}$ is non-singular. Since $\rank(\matV_{\matS\matA})=\rank(\matV_{\matA})$, it follows that the range of $\matV_{\matS\matA}$ equals to the range of $\matV_{\matA}$. So $\matV_{\matS\matA} = \matV_{\matA} \matW$ for some unitary matrix $\matW$ of size $p$. $\matV_{\matS\matA}\transp \matV_{\matA} = \matW^\top$ and $\matW$ is non-singular and so is 
$\matD_L$.
The second inequality follows because for any two matrices $\matX, \matY:$ $\sigma_i(\matX \matY) \le \TNorm{\matX} \sigma_i(\matY)$.
Finally, in the third inequality we used the fact that $\TNorm{ \matU_{\matA}\transp \matS\transp } \le \sqrt{1+\epsilon}$ and $\sigma_i\left( \matS \matU_{\matB} \right) \le \sqrt{1+\epsilon}.$
We now bound
$\TNorm{\matSig_{\matS\matA}^{-1} \matV_{\matS\matA}\transp \matV_{\matA} \matSig_{\matA}^2 \matV_{\matA}\transp \matV_{\matS\matA} \matSig_{\matS\matA}^{-1} - \matI_p}$.
(The second term in the max expression of $\gamma$ can be bounded in a similar fashion, so we omit the proof.)
\small
\begin{eqnarray*}
	 \TNorm{\matSig_{\matS\matA}^{-1} \matV_{\matS\matA}\transp \matV_{\matA} \matSig_{\matA}^2 \matV_{\matA}\transp \matV_{\matS\matA} \matSig_{\matS\matA}^{-1} - \matI_p}
	&= & \TNorm{\matSig_{\matS\matA}^{-1} \matV_{\matS\matA}\transp \matA\transp\matA  \matV_{\matS\matA} \matSig_{\matS\matA}^{-1} - \matI_p}\\
	&= & \TNorm{ \matU_{\matS\matA}\transp \tpinv{(\matS\matA)} \matA\transp \matA  \pinv{(\matS\matA )} \matU_{\matS\matA}
     - \matU_{\matS \matA}\transp \matU_{\matS \matA}\matU_{\matS \matA}\transp \matU_{\matS \matA}}\\
	&= & \TNorm{ \matU_{\matS\matA}\transp \left( \tpinv{(\matS\matA)} \matA\transp \matA  \pinv{(\matS\matA )} - \matU_{\matS \matA} \matU_{\matS \matA}\transp\right) \matU_{\matS\matA}}\\
	&\leq & \TNorm{\tpinv{(\matS\matA)} \matA\transp \matA  \pinv{(\matS\matA )} - \matU_{\matS \matA} \matU_{\matS \matA}\transp}\\
\end{eqnarray*}
\normalsize
where we used $\matA\transp \matA = \matV_{\matA} \matSig_{\matA}^2 \matV_{\matA}\transp,$ and
$\pinv{(\matS\matA )} \matU_{\matS\matA} = \matV_{\matS\matA} \matSig_{\matS\matA}^{-1}$.
Recall that, all the singular values of $\matS \matU_{\matA}$ are between $\sqrt{1 - \epsilon}$ and $\sqrt{1 + \epsilon}$, so:
$$ (1-\epsilon) \matI_p \preceq \matU_{\matA}\transp \matS\transp \matS \matU_{\matA} \preceq (1+\epsilon ) \matI_p.$$
%Conjugating the PSD ordering with $\matSig_{\matA} \matV_{\matA}\transp$ (see Lemma~\ref{lem:pert3}), it follows that
%\begin{equation}
%(1-\epsilon/2) \matA\transp \matA \preceq \matA\transp \matS\transp \matS \matA \preceq (1+\epsilon/2) \matA\transp \matA.
%\end{equation}
%Conjugating the PSD ordering with $\pinv{(\matS \matA)}$ (see Lemma~\ref{lem:pert3}), it follows that
%$$(1-\epsilon/2) \tpinv{(\matS\matA)} \matA\transp \matA  \pinv{(\matS\matA )}
%\preceq \matU_{\matS \matA} \matU_{\matS \matA}\transp \preceq (1+\epsilon/2) \tpinv{(\matS\matA)} \matA\transp \matA  \pinv{(\matS\matA )}$$
%since $\matS\matA \pinv{(\matS\matA)} = \matU_{\matS\matA} \matU_{\matS\matA}\transp$.
Conjugating the above PSD ordering with $\matSig_{\matA} \matV_{\matA}\transp\pinv{(\matS \matA)}$ (see Lemma~\ref{lem:pert3}), it follows that
\begin{eqnarray*}
 (1  - \epsilon) \tpinv{(\matS\matA)} \matA\transp \matA  \pinv{(\matS\matA )}
  \preceq  \matU_{\matS \matA} \matU_{\matS \matA}\transp  \preceq  (1+\epsilon) \tpinv{(\matS\matA)} \matA\transp \matA  \pinv{(\matS\matA )}
\end{eqnarray*}
since $\matU_\matA\transp \matU_\matA = \matI_p $ and $ (\pinv{(\matS\matA)})\transp \matA\transp\matS\transp \matS\matA \pinv{(\matS\matA)} = \matU_{\matS\matA} \matU_{\matS\matA}\transp$.
Rearranging terms, it follows that
\begin{eqnarray*}
	\frac1{1+\epsilon} \matU_{\matS \matA} \matU_{\matS \matA}\transp  & \preceq & \tpinv{(\matS\matA)} \matA\transp \matA  \pinv{(\matS\matA )}
	 \preceq  \frac1{1-\epsilon} \matU_{\matS \matA} \matU_{\matS \matA}\transp.
\end{eqnarray*}
Since $0 < \epsilon < 1/2$, it holds that $\frac1{1-\epsilon/3}  \leq 1 + 2\epsilon$ and $\frac1{1+\epsilon} \geq 1 - \epsilon,$ hence
\begin{eqnarray*}
-2\epsilon \matU_{\matS \matA} \matU_{\matS \matA}\transp
\preceq  \tpinv{(\matS\matA)} \matA\transp \matA  \pinv{(\matS\matA )} - \matU_{\matS \matA} \matU_{\matS \matA}\transp
 \preceq  2\epsilon \matU_{\matS \matA} \matU_{\matS \matA}\transp
\end{eqnarray*}
using standard properties of the PSD ordering. This implies that
\begin{eqnarray*}
 \TNorm{\tpinv{(\matS\matA)} \matA\transp \matA  \pinv{(\matS\matA )}  -  \matU_{\matS \matA} \matU_{\matS \matA}\transp}
  \leq  2\epsilon \TNorm{ \matU_{\matS \matA} \matU_{\matS \matA}\transp } = 2\epsilon \,.
\end{eqnarray*}
Indeed, let $\x_+$ be the unit eigenvector of the symmetric matrix
$$\tpinv{(\matS\matA)} \matA\transp \matA  \pinv{(\matS\matA )} - \matU_{\matS \matA} \matU_{\matS \matA}\transp$$ corresponding to its maximum eigenvalue. The PSD ordering implies that
\begin{eqnarray*}
\lambda_{\max}\left( \tpinv{(\matS\matA)} \matA\transp \matA  \pinv{(\matS\matA )} - \matU_{\matS \matA} \matU_{\matS \matA}\transp \right)
\leq  2 \epsilon \x_+\transp \matU_{\matS \matA} \matU_{\matS \matA}\transp \x_+ \leq 2\epsilon \TNorm{\matU_{\matS \matA} \matU_{\matS \matA}\transp}= 2\epsilon.
\end{eqnarray*}
Similarly,
$$\lambda_{\min}\left( \tpinv{(\matS\matA)} \matA\transp \matA  \pinv{(\matS\matA )} - \matU_{\matS \matA} \matU_{\matS \matA}\transp \right)> - 2\epsilon,$$
which shows the claim.
\end{proof}

\begin{lemma}\label{lem:pert6}
Repeat the conditions of Lemma~\ref{lem:pert4}.
Then, for all $\w \in \R^n$ and $\y \in \R^{\ell}$, we have
$$ \abs{ \w\transp \matA\transp \matB \y - \w\transp \matA\transp \matS\transp \matS \matB \y  }
\leq \epsilon \cdot \TNorm{\matA \w} \cdot \TNorm{\matB \y}.
$$
\end{lemma}
\begin{proof}
Let $\matE =  \matU_{\matA}\transp \matS\transp \matS \matU_{\matB} - \matU_{\matA}\transp \matU_{\matB}$. Now,
%$ \abs{ \w\transp \matA\transp \matB \y - \w\transp \matA\transp \matS\transp \matS \matB \y  } =$
\begin{eqnarray*}
\abs{ \w\transp \matA\transp \matB \y - \w\transp \matA\transp \matS\transp \matS \matB \y  } &=&    \abs{ \w\transp \matV_{\matA}\matSig_{\matA} \matE \matSig_{\matB} \matV_{\matB}\transp \y  }\\
&\le&  \TNorm{\w\transp \matV_{\matA}\matSig_{\matA}} \TNorm{\matE} \TNorm{\matSig_{\matB} \matV_{\matB}\transp \y}\\
&=&    \TNorm{\w\transp \matV_{\matA}\matSig_{\matA}\matU_{\matA}\transp} \TNorm{\matE} \TNorm{\matU_{\matB}\matSig_{\matB} \matV_{\matB}\transp \y}\\
&=&    \TNorm{\w\transp \matA\transp} \TNorm{\matE} \TNorm{\matB \y}\\
&=&    \TNorm{\matE}  \TNorm{\matA \w} \TNorm{\matB \y}
\end{eqnarray*}
Now, the proof of Lemma~\ref{lem:pert4} ensures that $\TNorm{\matE} \le \epsilon$.
\end{proof}
\section{CCA of Row Sampled Pairs}

Given $\matA$ and $\matB$, one straightforward way to accelerate CCA is to sample rows uniformly from both matrices, and to compute the CCA of the smaller matrices. %Most methods for computing the canonical correlations benefit from smaller matrices.
In this section we show that if we sample enough rows, then the canonical correlations of the sampled pair are close to the canonical correlations of the original pair. Furthermore, the canonical weights of the sampled pair can be used to find approximate canonical vectors. Not surprisingly, the sample size depends on the coherence. More specifically, it depends on the coherence of $[\matA ; \matB]$.

\begin{theorem}\label{thm1}
Suppose $\matA \in \R^{m \times n}$ ($m \geq n$) has rank $p$ and $\matB \in \R^{m \times \ell}$ ($m \geq \ell$) has rank $q \le p$. Let $0 < \epsilon < 1/2$ be an accuracy parameter and $0 < \delta < 1$ be a failure probability parameter. Let $\omega = \rank([\matA ; \matB]) \leq p+q$. Let $r$ be an integer such that
\[
54 \epsilon^{-2} m  \mu([\matA ; \matB]) \log ( 12 \omega /\delta) \leq r \leq m \,.
\]
Let $T$ be a random subset of $[m]$ of cardinality $r$, drawn from a uniform distribution over such subsets,
and let $\matS \in \R^{r \times m}$ be the sampling matrix corresponding to $T$ rescaled by $\sqrt{m/r}$.
Denote $\hat{\matA} = \matS \matA$ and $\hat{\matB} = \matS  \matB$.

Let $\hat{\sigma}_1,\dots,\hat{\sigma}_q$  be the exact canonical correlations of $(\hat{\matA}, \hat{\matB})$,
and let $$\w_1=\hat{\x}_1 / \TNorm{\hat{\matA} \hat{\x}_1}, \dots, \w_p=\hat{\x}_q / \TNorm{\hat{\matA} \hat{\x}_q}\,,$$ and
$$\p_1=\hat{\y}_1 / \TNorm{\hat{\matB} \hat{\y}_1}, \dots, p_q = \hat{\y}_q / \TNorm{\hat{\matB} \hat{\y}_q}$$ be the exact canonical weights of $(\hat{\matA}, \hat{\matB})$. With probability of at least $1-\delta$ all the following hold simultaneously:
\begin{enumerate}[(a)]
    \item
    (Approximation of Canonical Correlations)
    For every $i=1,2,\ldots ,q$: 
    
    $$ | \sigma_i\left(\matA, \matB \right) - \sigma_i\left( \hat\matA, \hat\matB \right)| \le  \epsilon + 2\epsilon^2 / 9 = O(\epsilon)\,.$$
	\item
	(Approximate Orthonormal Bases)
    The vectors $\{\matA \w_i\}_{i\in{[q]}}$ form an approximately orthonormal basis. That is,
    for any $c \in [q]$, $$\frac{1}{1+\epsilon/3} \leq \TNormS{\matA \w_c} \leq \frac{1}{1-\epsilon/3}\,,$$
	and for any $i\neq j$, \[|\ip{\matA \w_i}{ \matA \w_j}| \leq \frac{\epsilon}{3 - \epsilon}.\]
	Similarly, for the set of $\{\matB \p_i\}_{i\in{[q]}}$.
	\item
	(Approximate Correlation)
    For every $i=1,2,\ldots ,q$:
    \small
	$$ \frac{\sigma_i(\matA,\matB)}{1+\epsilon/3} - \frac{\epsilon/3}{1-\epsilon/9} \leq \sigma(\matA \w_i, \matB \p_i) \leq \frac{\sigma_i(\matA,\matB)}{1-\epsilon/3} + \frac{\epsilon/3}{(1-\epsilon/3)^2}\,.$$
\normalsize
\end{enumerate}
\end{theorem}
\begin{proof}
Let $\matC :=  [\matU_{\matA} ; \matU_{\matB}]$. Lemma~\ref{lemma:sampling-ortho} implies that each of the following three assertions hold with probability of at least $1-\delta/3$, hence all three events hold simultaneously with probability of at least $1-\delta$:
\begin{itemize}
\item For every $r\in[p]$: $\sqrt{1-\epsilon/3} \le \sigma_r(\matS \matU_{\matA}) \le  \sqrt{1+\epsilon/3}\,.$
\item For every $k\in[q]$: $\sqrt{1-\epsilon/3} \le \sigma_k(\matS \matU_{\matB}) \le  \sqrt{1+\epsilon/3}\,.$
\item For every $h\in[\omega]$: $\sqrt{1-\epsilon/3} \le \sigma_h(\matS \matU_{\matC}) \le  \sqrt{1+\epsilon/3}\,.$
\end{itemize}
We now show that if indeed all three events hold, then (a)-(c) hold as well.

{\bf Proof of (a).} Corollary~\ref{cor:bjork-golub} implies that $\sigma_i(\matA, \matB) = \sigma_i(\matU\transp_\matA \matU_\matB)$, and
$\sigma_i(\matS \matA, \matS \matB) = \sigma_i(\matU\transp_{\matS \matA} \matU_{\matS \matB})$. We now use the triangle inequality to get,
\begin{eqnarray*}
| \sigma_i \left( \matA, \matB \right) -  \sigma_i \left(\matS \matA, \matS\matB \right) |
    & = &  | \sigma_i \left( \matU_{\matA}\transp\matU_{\matB} \right) -   \sigma_i \left( \matU_{\matS\matA}\transp\matU_{\matS\matB} \right) |  \\
%    & = &  | \sigma_i \left( \matU_{\matA}\transp\matU_{\matB} \right)  - \sigma_i\left( \matU_{\matA}\transp \matS\transp \matS \matU_{\matB} \right) +  \sigma_i\left( \matU_{\matA}\transp \matS\transp \matS \matU_{\matB} \right) -   \sigma_i \left( \matU_{\matS\matA}\transp\matU_{\matS\matB} \right) | \\
    & \leq &  | \sigma_i \left( \matU_{\matA}\transp\matU_{\matB} \right)  - \sigma_i\left( \matU_{\matA}\transp \matS\transp \matS \matU_{\matB} \right)|
 	   +    |\sigma_i\left( \matU_{\matA}\transp \matS\transp \matS \matU_{\matB} \right) -   \sigma_i \left( \matU_{\matS\matA}\transp\matU_{\matS\matB} \right) |.
\end{eqnarray*}
To conclude the proof, use Lemma~\ref{lem:pert4} and Lemma~\ref{lem:pert5} to bound these two terms, respectively.

{\bf Proof of (b).} For any $c \in [q]$,
$$ \TNorm{\matA \w_c} = \TNorm{\matA \w_c} / \TNorm{\hat{\matA} \w_c}$$
since $\TNorm{\hat{\matA} \w_c} = 1$. Now Lemma~\ref{lem:pert6} implies the first inequality.

For any $i\neq j$
\small
\begin{eqnarray*}
	|\ip{\matA \w_i}{ \matA \w_j }| & \leq & |\w_i\transp \hat{\matA}\transp \hat{\matA}  \w_j|
	   +   |\w_i\transp (\hat{\matA}\transp \hat{\matA} - \matA\transp \matA)\w_j| \\
    &   =  & |\w_i\transp (\hat{\matA}\transp \hat{\matA} - \matA\transp \matA)\w_j| \\
	& \leq & \frac{\epsilon}{3} \TNorm{\matA \w_i} \TNorm{\matA \w_j} \\
	& \leq & \frac{\epsilon/3}{1-\eps/3} \TNorm{\hat{\matA}\w_i}\TNorm{\hat{\matA}\w_j}\\
	&   =  & \frac{\epsilon}{3-\eps}.
\end{eqnarray*}
\normalsize
In the above, we used the triangle inequality,
the fact that the $\w_i$'s are the canonical weights of $\hat{\matA}$,
and Lemma~\ref{lem:pert6}.

{\bf Proof of (c).} We only prove the upper bound. The lower bound is similar, and we omit it.\\
\small
%$\sigma\left(\matA \w_i, \matB \p_i\right) = $
\begin{eqnarray*}
\sigma\left(\matA \w_i, \matB \p_i\right)	 & = & \frac{\ip{\matA\w_i}{\matB\p_i}}{\TNorm{\matA\w_i}\TNorm{\matB\p_i}} \\
					& \leq & \frac{1}{1-\epsilon/3}\cdot\ip{\matA\w_i}{\matB\p_i}\\
                                              &   =  & \frac{1}{1-\epsilon/3}\cdot \left( \ip{\hat{\matA}\w_i}{\hat{\matB}\p_i}
                                                 +  \w\transp_i\left(\matA\transp\matB - \hat{\matA}\transp\hat{\matB} \right)\p_i \right) \\
                                              &   \leq  & \frac{\sigma\left(\hat{\matA} \x_i, \hat{\matB}\y_i\right)}{1-\epsilon/3}
                                                 +   \frac{\epsilon/3}{1-\epsilon/3}\cdot\TNorm{\matA\w_i}\cdot\TNorm{\matB\p_i} \\
											  &   \leq  & \frac{\sigma\left(\hat{\matA} \w_i, \hat{\matB}\p_i\right)}{1-\epsilon/3} + \frac{\epsilon/3}{(1-\epsilon/3)^2}
\end{eqnarray*}
\normalsize
In the above,
the first equality follows by the definition of $\sigma(\cdot,\cdot)$,
the first inequality by using $$1=\TNormS{\hat{\matA}\w_i} \leq (1+\eps) \TNormS{\matA\w_i},$$ (same holds for $\matB\p_i$),
the second inequality from Lemma~\ref{lem:pert6},
the third inequality  by using $$(1-\eps)\TNormS{\matA\w_i}\leq \TNormS{\hat{\matA} \w_i} =1,$$ (same holds for $\matB\p_i$),
and the last inequality by (a).
\end{proof}

\section{Fast Approximate CCA}\label{sec:alg}

First, we define what we mean by approximate CCA.
\begin{definition}[Approximate CCA] \label{def:approxCCA}
For $0 \leq \eta \leq 1$, an {\em $\eta$-approximate CCA of $(\matA, \matB)$}, is a set of positive numbers $\hat{\sigma}_1,\dots,\hat{\sigma}_q$ together with a set of vectors $\w_1,\dots,\w_q$ for $\matA$ and a set of vectors $\p_1,\dots,\p_q$ for $\matB$, such that
\begin{enumerate}[(a)]
\item For every $i\in[q]$, $$|\sigma_i(\matA, \matB) - \hat{\sigma}_i | \leq \eta\,.$$
\item For every $i\in[q]$, $$|\TNormS{\matA \w_i} - 1 | \leq \eta\,,$$ and for $i\neq j$, $$|\ip{\matA \w_i}{ \matA \w_j}| \leq \eta\,.$$ Similarly, for the set of $\{\matB \p_i\}_{i\in{[q]}}$.
\item For every $i\in[q]$, $$|\sigma_i(\matA, \matB) - \sigma(\matA \w_i, \matB \p_i) | \leq \eta\,.$$
\end{enumerate}
\end{definition}

We are now ready to present our fast algorithm for approximate CCA of a pair of tall-and-thin matrices. Algorithm~\ref{alg:approx} gives the pseudo-code description of our algorithm.

The analysis in the previous section (Theorem~\ref{thm1}) shows that if we sample enough rows, the canonical correlations and weights of the sampled matrices are an $O(\epsilon)$-approximate CCA of $(\matA, \matB)$. However, to turn this observation into a concrete algorithm we need an upper bound on the coherence of $[\matA ; \matB]$. It is conceivable that in certain scenarios such an upper bound might be known in advance, or that it can be computed quickly~\cite{DMMW12}. However, even if we know the coherence, it might be as large as one, which will imply that sampling the entire matrix is needed.

To circumvent this problem, our algorithm uses the RHT to reduce the coherence of the matrix pair before sampling rows from it. That is, instead of sampling rows from $(\matA, \matB)$  we sample rows from $(\matTh \matA, \matTh\matB)$, where $\matTh$ is a RHT matrix (Definition~\ref{def:rht}). This unitary transformation bounds the coherence with high probability, so we can use Theorem~\ref{thm1} to compute the number of rows required for an $O(\epsilon)$-approximate CCA.
We now sample the transformed pair $(\matTh \matA, \matTh\matB)$ to obtain $(\hat{\matA}, \hat{\matB})$. Now the canonical correlations and weights of $(\hat{\matA}, \hat{\matB})$ are computed and returned.

\begin{algorithm}[t]
\caption{Fast Approximate CCA}
\label{alg:approx}
\begin{algorithmic}[1]
\STATE {\bf Input:} $\matA \in \R^{m \times n}$ of rank $p$, $\matB \in \R^{m \times \ell}$ of rank $q$, $0< \epsilon < 1/2$, and $\delta$ ($n\geq l$, $p \ge q$).
\medskip
\STATE $r \longleftarrow \min(54\epsilon^{-2}\left[\sqrt{n+\ell} + \sqrt{8\log(12m/\delta)} \right]^2 \log (3(n+\ell)/\delta), m)$
\STATE Let $\matS$ be the sampling matrix of a random subset of $[m]$ of cardinality $r$ (uniform distribution).
\STATE Draw a random diagonal matrix $\matD$ of size $m$ with $\pm 1$ on its diagonal with equal probability.
\STATE $\hat{\matA} \longleftarrow \matS \matH \cdot (\matD \matA)$ using fast subsampled WHT (see Section~\ref{sec:wht}).
\STATE $\hat{\matB} \longleftarrow \matS \matH \cdot (\matD \matB)$ using fast subsampled WHT (see Section~\ref{sec:wht}).
\STATE Compute and return the canonical correlations and the canonical weights of $( \hat{\matA},\hat{\matB} )$
(e.g. using Bj{\"o}rck and Golub's algorithm).
\end{algorithmic}
\end{algorithm}

\begin{theorem}\label{thm:alg}
With probability of at least $1-\delta$, Algorithm~\ref{alg:approx} returns an $O(\epsilon)$-approximate CCA of $(\matA, \matB)$. Assuming Bj{\"o}rck and Golub's algorithm is used in line 7, Algorithm~\ref{alg:approx} runs in time
$$O\left(  m n \log{m} +  \epsilon^{-2}\left[\sqrt{n} + \sqrt{\log(m/\delta)}\right]^2 \log(n/\delta) n^2\right)\,.$$
\end{theorem}

\begin{proof}
%First we argue about the quality of approximation of Algorithm~\ref{alg:approx} and then about its time complexity.
Lemma~\ref{lem:rht-reduce} ensures that with probability of at least $1-\delta/2$,
$$\mu([\matTh \matA; \matTh \matB]) \leq \frac{1}{m} \left( \sqrt{n+\ell} + \sqrt{8\log(3m/\delta)} \right)^2\,.$$
Assuming that the last inequality holds, Theorem~\ref{thm1} ensures that with probability of at least $1-\delta/2$,
the canonical correlations and weights of $(\hat{\matA},  \hat{\matB})$ form an $O(\epsilon)$-approximate CCA of
$(\matTh \matA, \matTh \matB)$. By the union bound, both events hold together with probability of at least $1-\delta$.
The RHT transforms applied to $\matA$ and $\matB$ are unitary, so for every $\eta$, an $\eta$-approximate CCA of $(\matTh \matA, \matTh \matB)$ is also an $\eta$-approximate CCA of $(\matA, \matB)$ (and vice versa).

{\bf Running time analysis.}
Step 2 takes $O(1)$ operations. Step 3 requires $O(r)$ operations.
Step 4 requires $O(m)$ operations.
Step 5 involves the multiplication of $\matA$ with $\matS \matH \matD$ from the left.
Computing $\matD \matA$ requires $O(mn)$ time. Multiplying $\matS \matH$ by $\matD \matA$ using fast subsampled WHT requires
$O(m n \log r )$ time, as explained in Section~\ref{sec:wht}.
Similarly, step 6 requires $O( m \ell \log r )$ operations.
Finally, step 7 takes $O( r n \ell + r (n^2 + \ell^2))$ time.  Assuming that $n \geq \ell$, the total running time is
$O(r n^2 + m n \log(r))$. Plugging the value for $r$, and using the fact that $r \leq m$, establishes our running time bound.
\end{proof}

\section{Fast Approximate CCA with Other Transforms}\label{sec:sparse}

Our discussion so far has focused on the case in which we reduce the dimensions of $\matA$ and $\matB$ via the SRHT.
In recent years several similar transforms have been suggested by various researchers. 
For example, one can use the Fast Johnson-Lindenstraus method of Ailon and Chazelle~\cite{AC06}.
% or select rows from $\matA$ and $\matB$ based on leverage scores sampling~\cite{DMMW12}.
This transform leads to an approximate CCA algorithm
with a similar additive error gaurantee and running time as in Theorem~\ref{thm:alg}. 
%(We omit the details of the corresponding proof.)

Recently, Clarkson and Woodruff described a transform that is particularly appealing if the input matrices 
$\matA$ and $\matB$ are sparse~\cite{CW12}. We present this transform in the following lemma along with 
theoretical guarantees similar to those of Lemma~\ref{lemma:sampling-ortho}.  The following lemma is due 
to Meng and Mahoney~\cite{MM13}, which analyzed the transform originally due to Clarkson and Woodruff~\cite{CW12}.
We only employ the lemma due to Meng and Mahoney~\cite{MM13} because it slightly improves upon the original result due to Clarkson and Woodruff~\cite[Theorem 19]{CW12}.
\begin{lemma}\label{lem:Mahoney}[Theorem 1 in~\cite{MM13} with $\epsilon, \delta$ replaced with $\epsilon/3, \delta/3,$ respectively.]
Given any matrix $\matX \in \R^{m \times d}$ with $m \gg d$, accuracy parameter $0 < \epsilon < 1/3$, and failure probability parameter $0 < \delta < 1,$
let $$r \geq \ceil{ \frac{ 243 (d^2 + d) }{\epsilon^2 \delta}}.$$
Construct an $r \times m$ matrix $\matOmega$ as follows: $\matOmega = \matS \matD$, where $\matS \in \R^{r \times m}$ has each column chosen independently
and uniformly from the $r$ standard basis vectors of $\R^{r}$ and $\matD \in \R^{m \times m}$ is a diagonal matrix with diagonal entries chosen independently and uniformly from $\{+1,-1\}$. Then with probability at least $1-\delta/3$, for every $j\in[d]$:
 $$\sqrt{1-\epsilon/3}  \cdot \sigma_j(\matX) \le \sigma_j(\matOmega \matX) \le  \sqrt{1+\epsilon/3} \cdot \sigma_j(\matX) \,.$$
Moreover, $\matOmega \matX$ can be calculated in $O(\nnz(\matX))$ arithmetic operations.
\end{lemma}

\begin{algorithm}[t]
\caption{Fast Approximate CCA with the Clarskon-Woodruff Transform~\cite{CW12}}
\label{alg:approx2}
\begin{algorithmic}[1]
\STATE {\bf Input:} $\matA \in \R^{m \times n}$ of rank $p$, $\matB \in \R^{m \times \ell}$ of rank $q$, $0< \epsilon < 1/3$, and $\delta$ ($n\geq l$, $p \ge q$).
\medskip
\STATE $r \longleftarrow \min\left( \frac{ 243 \left( \left(n+\ell \right)^2 + n+\ell \right) }{\epsilon^2 \delta}, m \right)$.
\STATE Let $\matS$ be an $r \times m$ matrix constructed as follows: it has each column chosen independently
and uniformly from the $r$ standard basis vectors of $\R^{r}$
\STATE Draw a random diagonal matrix $\matD$ of size $m$ with $\pm 1$ on its diagonal with equal probability.
\STATE $\hat{\matA} \longleftarrow \matS \cdot (\matD \matA)$.
\STATE $\hat{\matB} \longleftarrow \matS \cdot (\matD \matB)$.
\STATE Compute and return the canonical correlations and the canonical weights of $( \hat{\matA},\hat{\matB} )$
(e.g. using Bj{\"o}rck and Golub's algorithm).
\end{algorithmic}
\end{algorithm}
Similarly to Theorem~\ref{thm1} we have the following theorem.
\begin{theorem}\label{thmCW}
Suppose $\matA \in \R^{m \times n}$ ($m \geq n$) has rank $p$ and $\matB \in \R^{m \times \ell}$ ($m \geq \ell$) has rank $q \le p$. Let $0 < \epsilon < 1/3$ be an accuracy parameter and $0 < \delta < 1$ be a failure probability parameter. Let $\omega = \rank([\matA ; \matB]) \leq p+q$. Let $r$ be an integer such that
\[
\frac{ 243 (\omega^2 + \omega) }{\epsilon^2 \delta} \leq r \leq m \,.
\]
Let $\matOmega \in \R^{r \times m}$ be constructed as in Lemma~\ref{lem:Mahoney}. Denote $\hat{\matA} = \matOmega \matA$ and $\hat{\matB} = \matOmega  \matB$.

Let $\hat{\sigma}_1,\dots,\hat{\sigma}_q$  be the exact canonical correlations of $(\hat{\matA}, \hat{\matB})$,
and let $\w_1=\hat{\x}_1 / \TNorm{\hat{\matA} \hat{\x}_1},$ $\dots,\w_p=\hat{\x}_q / \TNorm{\hat{\matA} \hat{\x}_q}\,,$ and
$\p_1=\hat{\y}_1 / \TNorm{\hat{\matB} \hat{\y}_1}, \dots, p_q = \hat{\y}_q / \TNorm{\hat{\matB} \hat{\y}_q}$ be the exact canonical weights of $(\hat{\matA}, \hat{\matB})$. With probability of at least $1-\delta$ all three statements (a), (b), and (c) of Theorem~\ref{thm1} hold simultaneously.
%\begin{enumerate}[(a)]
%    \item
%    (Approximation of Canonical Correlations)
%    For every $i=1,2,\ldots ,q$: $ | \sigma_i\left(\matA, \matB \right) - \sigma_i\left( \matOmega \matA, \matOmega  \matB \right)| \le  \epsilon + 2\epsilon^2 / 9 = O(\epsilon)\,.$
%	\item
%	(Approximate Orthonormal Bases)
%    The vectors $\{\matA \w_i\}_{i\in{[q]}}$ form an approximately orthonormal basis. That is,
%    for any $c \in [q]$, $$\frac{1}{1+\epsilon/3} \leq \TNormS{\matA \w_c} \leq \frac{1}{1-\epsilon/3}\,,$$
%	and for any $i\neq j$, \[|\ip{\matA \w_i}{ \matA \w_j}| \leq \frac{\epsilon}{3 - \epsilon}.\]
%	Similarly, for the set of $\{\matB \p_i\}_{i\in{[q]}}$.
%	\item
%	(Approximate Correlation)
%    For every $i=1,2,\ldots ,q$:
%    \small
%	$$ \frac{\sigma_i(\matA,\matB)}{1+\epsilon/3} - \frac{\epsilon/3}{1-\epsilon/9} \leq \sigma(\matA \w_i, \matB \p_i) \leq \frac{\sigma_i(\matA,\matB)}{1-\epsilon/3} + \frac{\epsilon/3}{(1-\epsilon/3)^2}\,.$$
%\normalsize
%\end{enumerate}
\end{theorem}
\begin{proof}
Let $\matC :=  [\matU_{\matA} ; \matU_{\matB}]$. Lemma~\ref{lem:Mahoney} implies that each of the following three assertions hold with probability of at least $1-\delta/3$, hence all three hold simultaneously with probability of at least $1-\delta$:
\begin{itemize}
\item For every $r\in[p]$: $\sqrt{1-\epsilon/3} \le \sigma_r(\matOmega \matU_{\matA}) \le  \sqrt{1+\epsilon/3}\,.$
\item For every $k\in[q]$: $\sqrt{1-\epsilon/3} \le \sigma_k(\matOmega \matU_{\matB}) \le  \sqrt{1+\epsilon/3}\,.$
\item For every $h\in[\omega]$: $\sqrt{1-\epsilon/3} \le \sigma_h(\matOmega \matU_{\matC}) \le  \sqrt{1+\epsilon/3}\,.$
\end{itemize}
Recall that in the proof of Theorem~\ref{thm1} we have shown that if indeed all three hold, then (a)-(c) hold as well.
\end{proof}

Finally, similarly to Theorem~\ref{thm:alg} we have the following theorem for approximate CCA (see also Algorithm~\ref{alg:approx2}).
\begin{theorem}\label{thm:alg2}
With probability of at least $1-\delta$, Algorithm~\ref{alg:approx2} returns an $O(\epsilon)$-approximate CCA of $(\matA, \matB)$.
Assuming Bj{\"o}rck and Golub's algorithm is used in line 7, Algorithm~\ref{alg:approx2} runs in time
$$O\left( m + \nnz(\matA) + \nnz(\matB) +  n^4 \epsilon^{-2} \delta^{-1}  \right).$$
\vspace{-0.2in}
\end{theorem}
\begin{proof}
The bound is immediate from Theorem~\ref{thmCW} since $n + \ell \ge \omega$. So, we only need to analyze the running time.
Step 2 takes $O(1)$ operations. Step 3 requires $O(m)$ operations.
Step 4 requires $O(m)$ operations as well.
Step 5 involves the multiplication of $\matA$ with $\matS  \matD$ from the left. Lemma~\ref{lem:Mahoney} argues that this can be accomplished in $O( \nnz(\matA))$ arithmetic operations. Similarly, step 6 requires $O( \nnz(\matB))$ operations.
Finally, step 7 takes $O( r n \ell + r (n^2 + \ell^2))$ arithmetic operations.  Assuming that $n \geq \ell$, the total running time is
$O\left( m+\nnz(\matA)+\nnz(\matB) +   r n^2 \right)$. Plugging the value for $r= \left(243 \left( \left(n+\ell \right)^2 + n+\ell \right) \right) / \left(\epsilon^2 \delta \right)$
and using again that $n \ge \ell$ establishes the bound.
\end{proof}

\subsection*{Sufficient properties of a dimension reduction transform}
We  stress that the three bounds stated in the beginning of the proof of Theorem~\ref{thmCW} are
three sufficient conditions for \emph{any} matrix $\matOmega$ one would like to pick and design a dimensionality reduction algorithm for CCA with provable guarantees. 
%Having said that, the Fast Johnson-Lindenstrauss transform~\cite{AC06} is known to achieve similar bounds for a sufficiently large value of $r$.

%%%%%%%%%%%%%%%%%%%%%%%%%%%%%%%%%%%%%%%%%%%%%%%%%%%%%%%%%%%%%%%%%%
% Optimality of bounds
%%%%%%%%%%%%%%%%%%%%%%%%%%%%%%%%%%%%%%%%%%%%%%%%%%%%%%%%%%%%%%%%%%
\section{Relative vs. Additive Error}\label{sec:error:lowerbound}
% Haim: made this a section
%%%%%%%%%%%%%%%%%%%%%%%%%%%%%%%%%%%%%%%%%%%%%%%%%%%%%%%%%%%%%%%%%%
In this section we prove that it is not possible to replace the additive error guarantees of Theorem~\ref{thm:alg} with relative error guarantees unless $r \approx m$. To prove such a statement we leverage tools from communication complexity~\cite{book:CC,CC:Yao}.
In general, communication complexity studies the following problem involving two parties (usually referred as Alice and Bob). Alice and Bob privately receive an $m$-bit string $\x$ and an $m$-bit string $\y$, respectively. The goal is to compute a certain function $f(\x,\y)$ with the least amount of communication (in bits) between them. We are assuming that they both follow a predefined communication protocol $\mathcal{P}$ agreed upon beforehand. The protocol consists of the players sending bits to each other until the value of $f$ can be determined, see~\cite{book:CC} for more details. Probabilistic protocols in which players have access to random bits (coin tosses) can be also defined\footnote{There are two models depending on whether the coin tosses are public or private. In the public random string model the players share a common random bit-string, while in the private model each player has his/her own private random bit-string. Here we focus on the public model.}. We say that a randomized protocol $\mathcal{P}$ computes a function $f$ with error $\delta$ if $\forall \x\in{\{0,1\}^m},\forall \y\in{\{0,1\}^m}:\ \Prob[ \mathcal{P}(\x,\y) = f(\x,\y)] \geq 1-\delta$. For $0<\delta<1/2$, $R_\delta(f)$ is the minimum worst case communication cost (in bits) over all randomized protocols that compute $f$ with error $\delta$.
In the proof we use a reduction to the set disjointness problem~\cite{CC:story}. The set disjointness problem $\text{DISJ}(\x,\y)$ is defined as follows: Alice gets an $\x\in\{0,1\}^m$ as input and Bob gets $\y\in\{0,1\}^m$. Their goal is to decide if there exists $i\in{[m]}$ so that $x_i = y_i = 1$ by exchanging as less information as possible. It is known that $R_{\delta}(\text{DISJ})=\Omega(m)$ for any constant $0<\delta<1/2$, see~\cite{DISJ},\cite[Theorem~17]{CC:story} for a modern proof. In the following lemma, we use the lower bound of the set disjointness problem to show that achieving relative error approximation for CCA (via using the SRHT specifically) while significantly reducing the dimensionality is impossible.
\begin{lemma}
Assume that given any matrix pair ($\matA$, $\matB$) and any constant $0<\eps <1$, Algorithm~\ref{alg:approx} computes a pair $(\hat{\matA}, \hat{\matB})$ by setting a sufficient large value for $r$ in Step $2$ so that the canonical correlations are relatively preserved with constant probability, i.e., with constant probability:
\begin{equation}\label{eq:relative}
	(1-\eps) \sigma_i(\matA,\matB) \leq \sigma_i(\hat{\matA},\hat{\matB}) \leq (1+\eps) \sigma_i(\matA,\matB) \quad i=1,\ldots ,q.
\end{equation}
Then, it follows that $r = \Omega( m/\log(m))$.
\end{lemma}
\begin{proof}
The proof follows by a reduction to the set disjointness communication complexity problem. 
% Haim: commented out since this just repeats what is said earlier.
That is, assume that Alice gets an $\x\in\{0,1\}^m$ as input and Bob gets $\y\in\{0,1\}^m$ (both $\x,\y$ are non-zero). Their goal is to decide if there exists $i\in{[m]}$ so that $x_i = y_i = 1$.
Set $\epsilon = 1/2$ and let $0<\delta<1/2$ be a constant in Algorithm~\ref{alg:approx}. Now, we will describe a protocol that solves the set disjointness problem using Algorithm~\ref{alg:approx} for the special case of two one dimensional subspaces. Alice and Bob can compute $\widetilde{\x} = \sqrt{m}\matS \matH\matD \x$ and $\widetilde{\y} = \sqrt{m}\matS \matH\matD \y$, respectively (using shared randomness). Then, Alice sends $\widetilde{\x}$ to Bob. Under the hypothesis (Eqn.~\ref{eq:relative}), it holds that
	\[ \frac1{2}\frac{\ip{\x}{\y}}{\norm{\x}\norm{\y}} \leq \frac1{r} \frac{\ip{\widetilde{\x}}{\widetilde{\y}}}{\norm{\widetilde{\x}}\norm{\widetilde{\y}}} \leq \frac{3}{2} \frac{\ip{\x}{\y}}{\norm{\x}\norm{\y}}, \]
with constant probability since $\sigma_1(\x,\y) = \frac{\ip{\x}{\y}}{\norm{\x}\norm{\y}}$. Now, Bob can decide if there exists $i$, so that $x_i=y_i =1$ by checking if $\ip{\widetilde{\x}}{\widetilde{\y}}$ is zero or non-zero. Hence, this protocol decides the set disjointness problem with constant probability $\delta$. Now, since $\sqrt{m}\matS\matH \matD$ is an $r\times m$ matrix with entries from $\{-1,+1\}$ and $\x\in\{0,1\}^m$, it follows that $\widetilde{\x}$ is integer-valued with $\infnorm{\widetilde{\x}} \leq m$. Therefore, we can encode $\widetilde{\x}$ using $r \log(2m)$ bits. Since $R_\delta(\text{DISJ})= \Omega (m)$, the number of bits exchanged between Alice and Bob must be at least $Cm$ for some constant $C>0$. Therefore $r\log(2m) \geq C m$.
\end{proof}

\section{Experiments}\label{sec:experiments}
In this section we report the results of a few small-scale experiments.
Our experiments are not meant to be exhaustive. However they do show that our algorithm
can be modified slightly to achieve very good performance in practice while still producing acceptable results.

Our implementation of Algorithm~\ref{alg:approx} differs from the pseudo-code description in two ways. First, we use
$$r \longleftarrow \min(\epsilon^{-2}\left[\sqrt{n+\ell} + \sqrt{\log(m/\delta)} \right]^2 \log (n+\ell)/\delta), m)$$
for setting the sample size, i.e., we keep the same asymptotic behavior, but drop the constants. The constants in Algorithm~\ref{alg:approx} are rather large, so they preclude the possibility of beating Bj{\"o}rck and Golub's algorithm for reasonable matrix sizes. Our implementation
also differs in the choice of the underlying mixing matrix. Algorithm~\ref{alg:approx}, and the analysis, uses the WHT. 
However, as we discussed in Section~\ref{sec:wht}, other Fourier-type transforms will work as well 
and some of these alternative transforms have certain advantages that make them better suited for an actual implementation~\cite{AMT10}.
Specifically, we use the implementation of randomized Discrete Hartley Transform in the Blendenpik library~\cite{AMT10}\footnote{Available at http://www.mathworks.com/matlabcentral/fileexchange/25241-blendenpik.}.

We report the results of three experiments. In each experiment we run our code five times on a fixed pair of matrices (datasets) $\matA$ and $\matB$, and compare the different outputs to the true canonical correlations. The first two experiments involved synthetic datasets, for which we set $\epsilon=0.25$ and $\delta=0.05$. The last experiment was conducted on a real-life dataset, and we used $\epsilon=0.5$ and $\delta=0.2$. All experiments were conducted in a 64-bit version of MATLAB $7.9$. We used a Lenovo W520 Thinkpad: Intel Corei7-2760QM CPU running at  2.40 GHz, with 8GB RAM, running Linux 3.5.The measured running times are wall-clock times and were measured using the {\bf ftime} Linux system call.
\begin{figure*}[t]
\noindent \begin{centering}
\begin{tabular}{ccc}
\includegraphics[width=0.27\columnwidth]{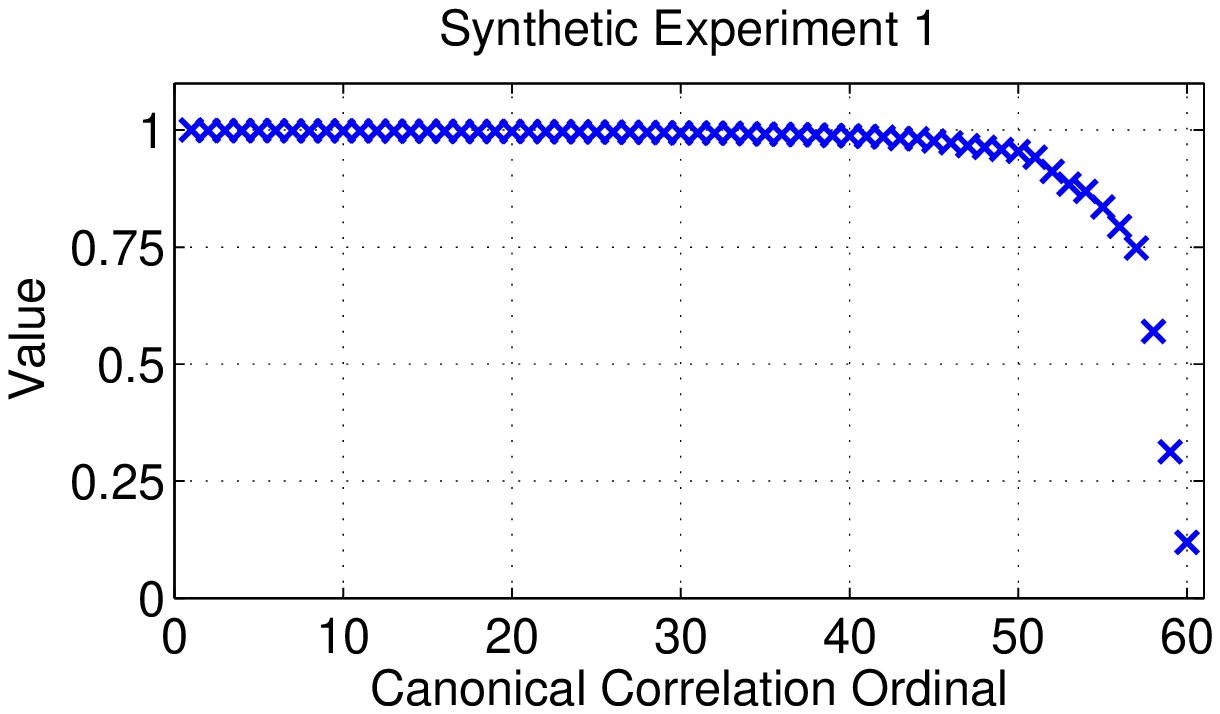} & \includegraphics[width=0.27\columnwidth]{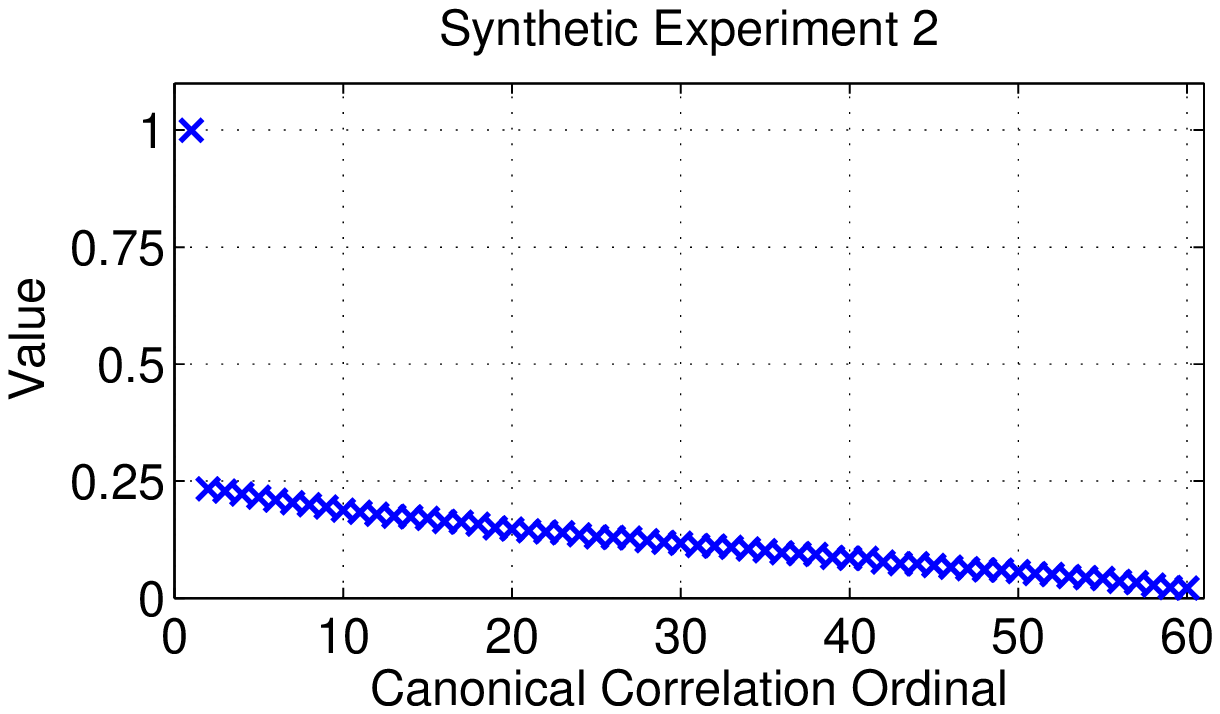} &
\includegraphics[width=0.27\columnwidth]{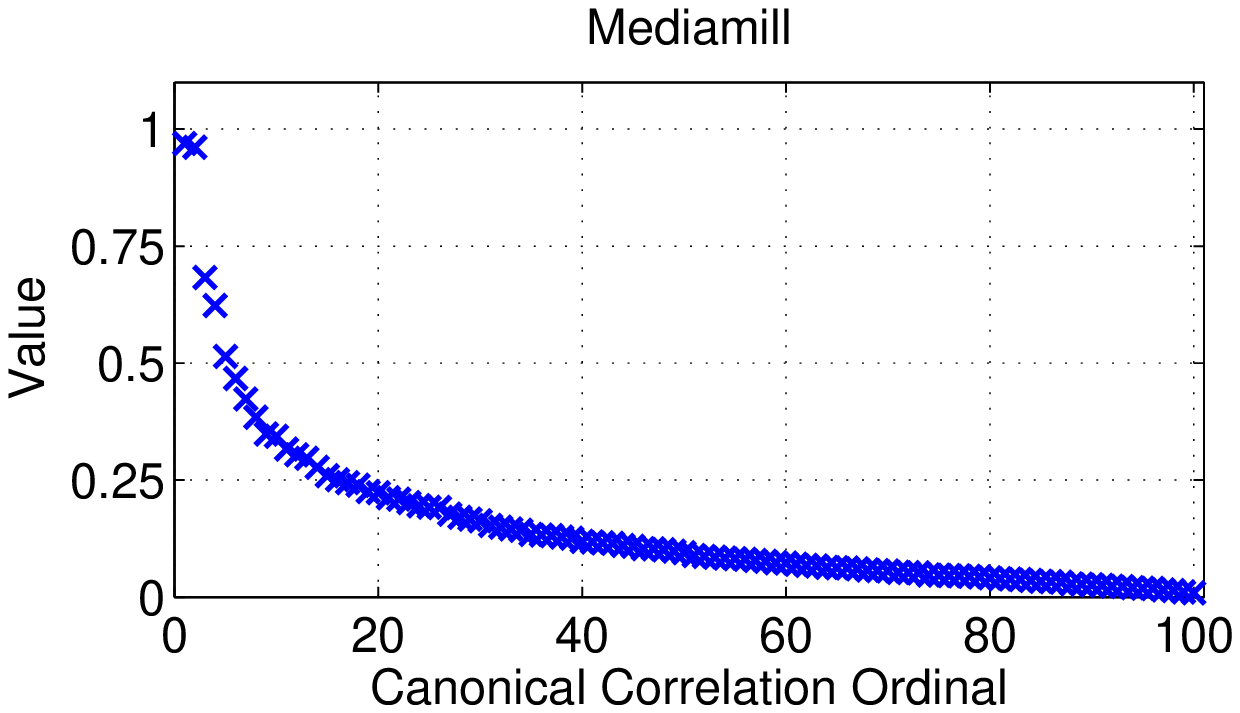} \tabularnewline
% for arxiv upload
%\includegraphics[width=0.27\columnwidth]{example1_cc} & \includegraphics[width=0.27\columnwidth]{example2_cc} &
%\includegraphics[width=0.27\columnwidth]{mediamill_cc} \tabularnewline
\textbf{(a)} & \textbf{(b)} & \textbf{(c)} \tabularnewline
\end{tabular}
\par\end{centering}
\caption{\label{fig:cc-values}The exact canonical correlations.}
%\end{figure*}
\vspace*{0.1in}
%\begin{figure*}[t]
\noindent \begin{centering}
\begin{tabular}{ccc}
\includegraphics[width=0.27\columnwidth]{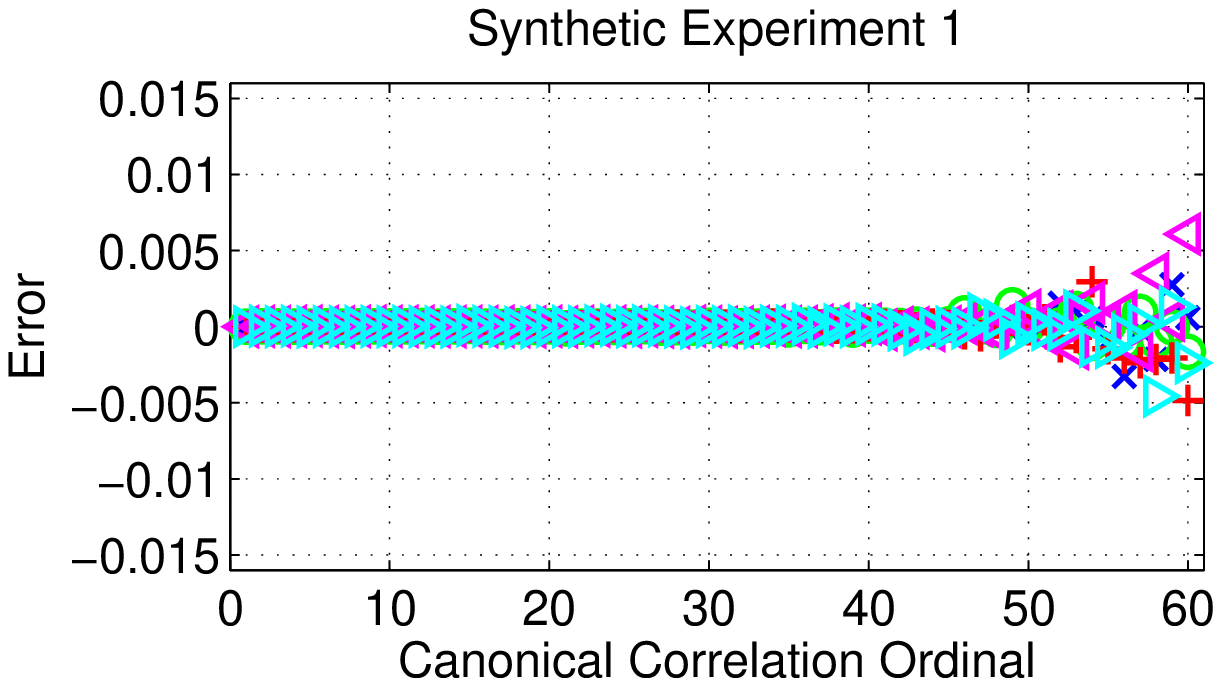} & \includegraphics[width=0.27\columnwidth]{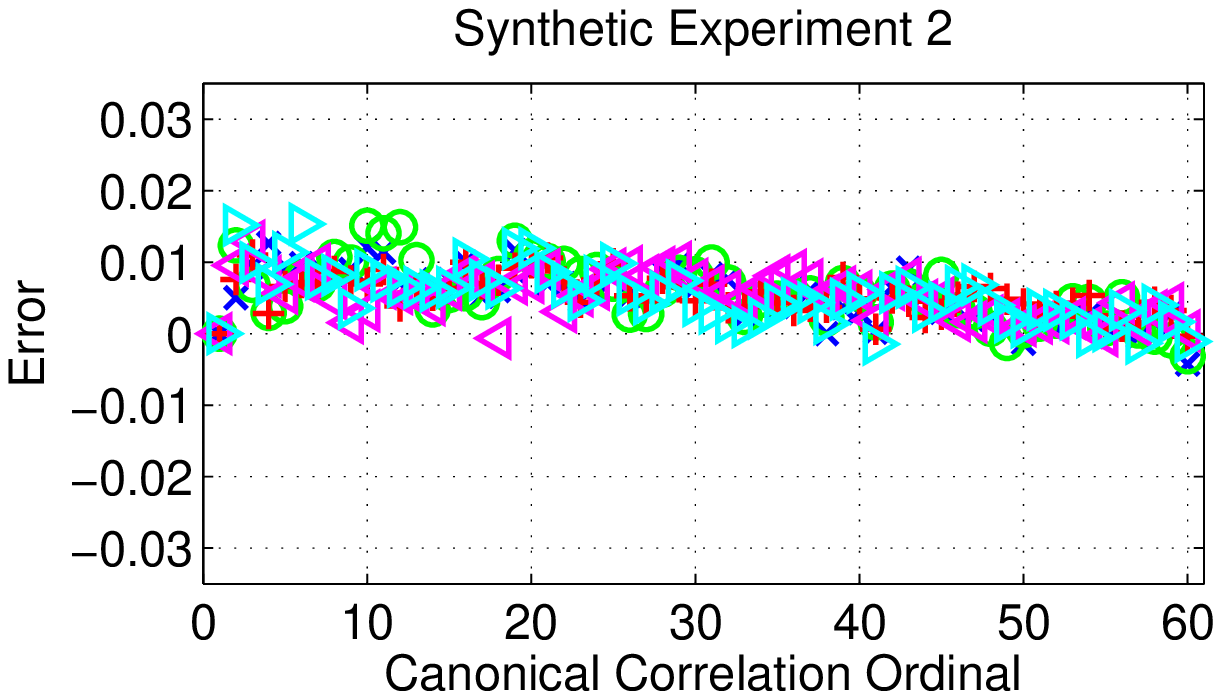} &
\includegraphics[width=0.27\columnwidth]{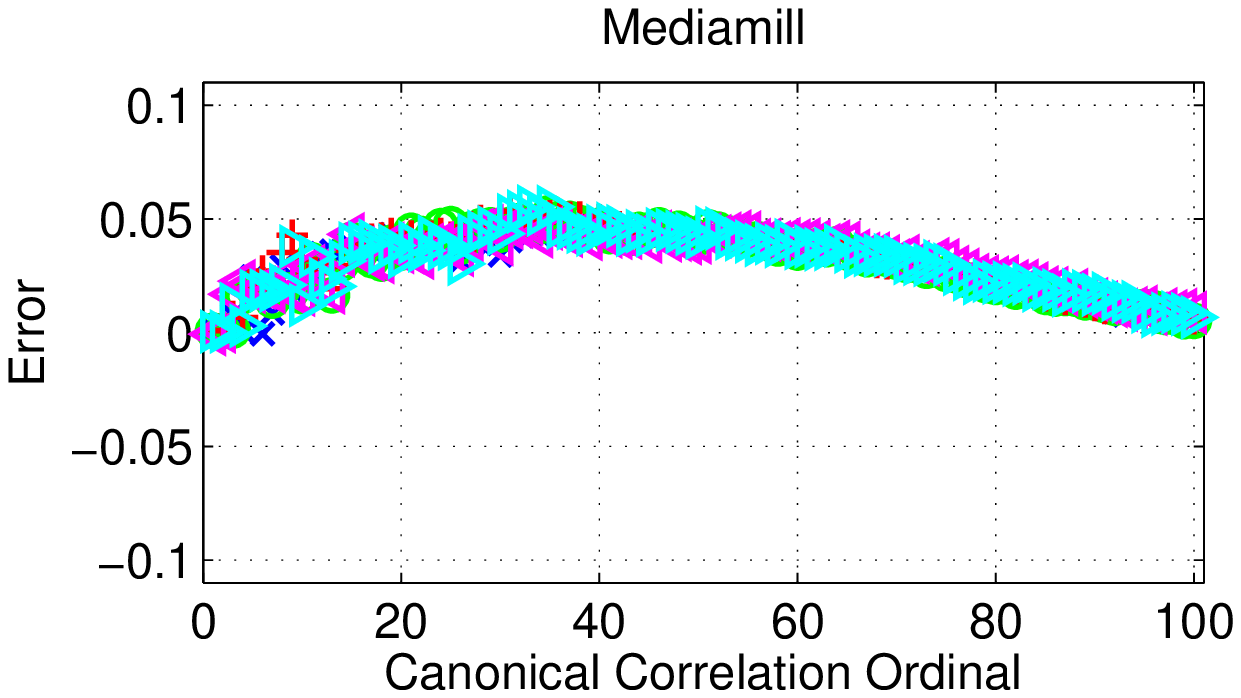}\tabularnewline
% for arxiv upload
%\includegraphics[width=0.27\columnwidth]{example1} & \includegraphics[width=0.27\columnwidth]{example2} &
%\includegraphics[width=0.27\columnwidth]{mediamill}\tabularnewline
\textbf{(a)} & \textbf{(b)} & \textbf{(c)} \tabularnewline
\end{tabular}
\par\end{centering}
\caption{\label{fig:cc-approx}Error in approximation of the canonical correlations.}
%\begin{figure*}[t]
\vspace*{0.1in}
\noindent \begin{centering}
\begin{tabular}{ccc}
\includegraphics[width=0.27\columnwidth]{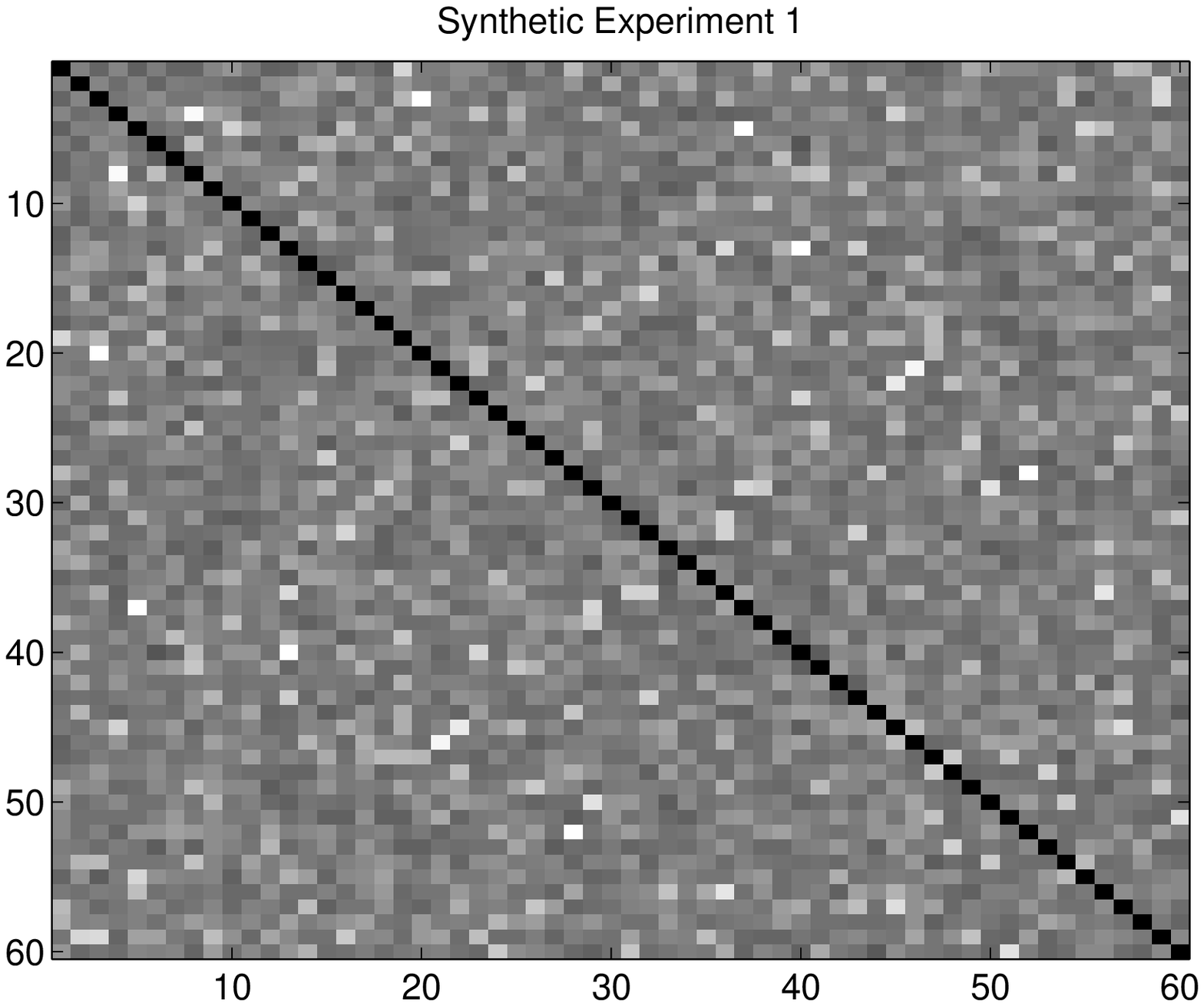} & \includegraphics[width=0.27\columnwidth]{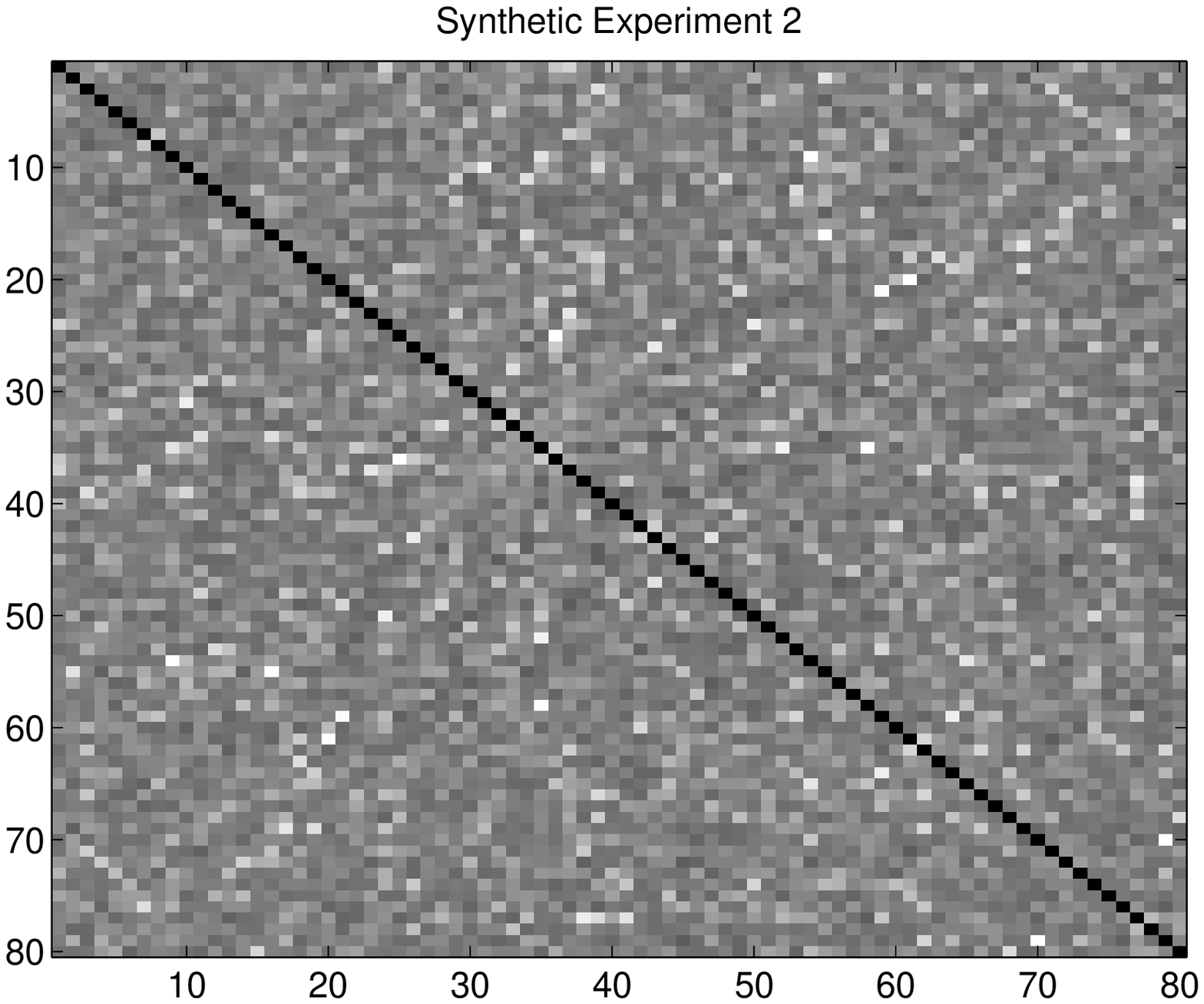} &
\includegraphics[width=0.27\columnwidth]{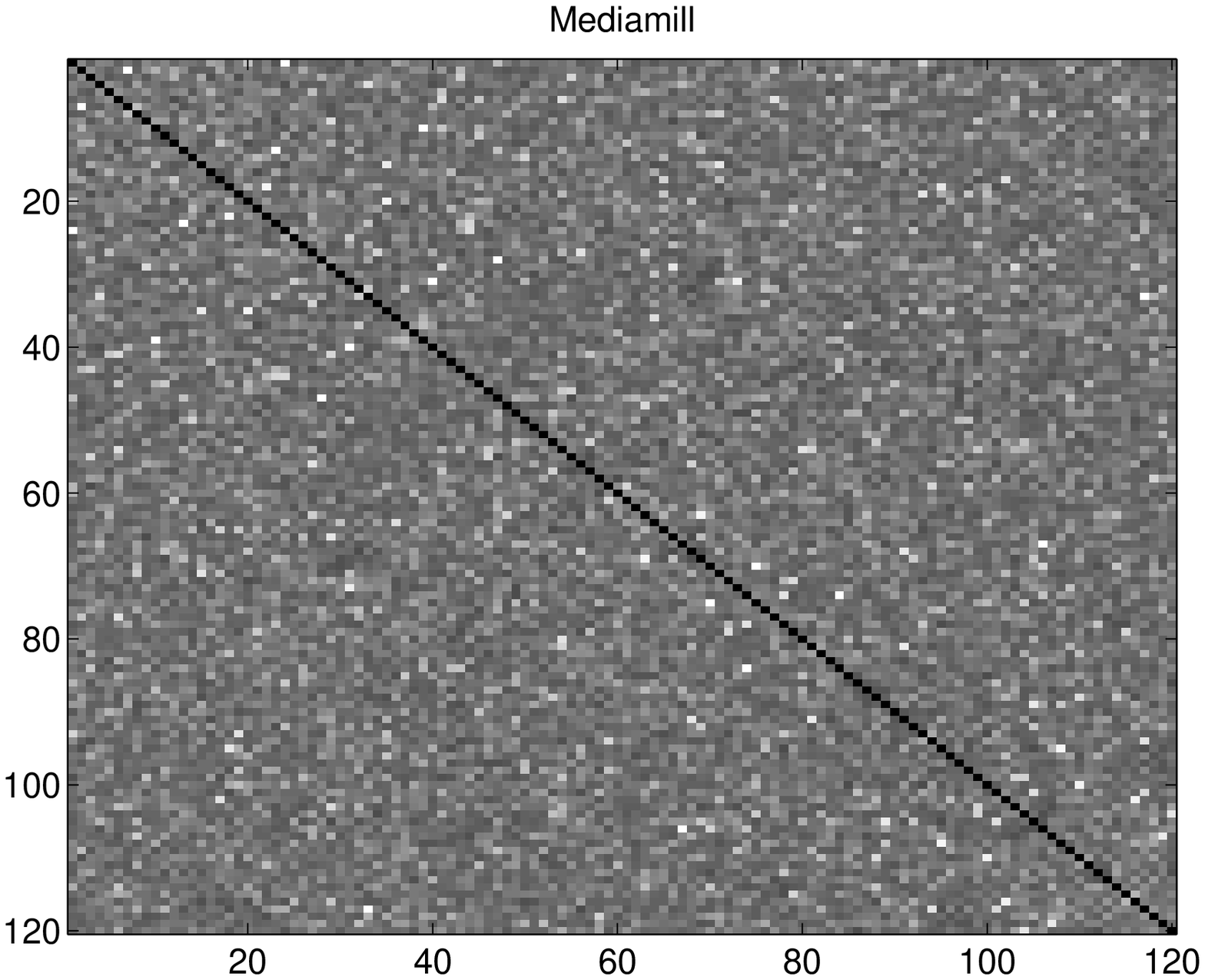}\tabularnewline
% for arxiv upload
%\includegraphics[width=0.27\columnwidth]{example1_aw} & \includegraphics[width=0.27\columnwidth]{example2_aw} &
%\includegraphics[width=0.27\columnwidth]{mediamill_aw}\tabularnewline
\textbf{(a)} & \textbf{(b)} & \textbf{(c)} \tabularnewline
\end{tabular}
\par\end{centering}
\caption{\label{fig:cc-aw}Visualization of the absolute value of the enteries in $\matW\transp \matA\transp \matA \matW$ in one of the runs. Color vary between white and black, with black is $1$ and white is $10^{-5}$.}
%\begin{figure*}[t]
\vspace*{0.1in}
\noindent \begin{centering}
\begin{tabular}{ccc}
\includegraphics[width=0.27\columnwidth]{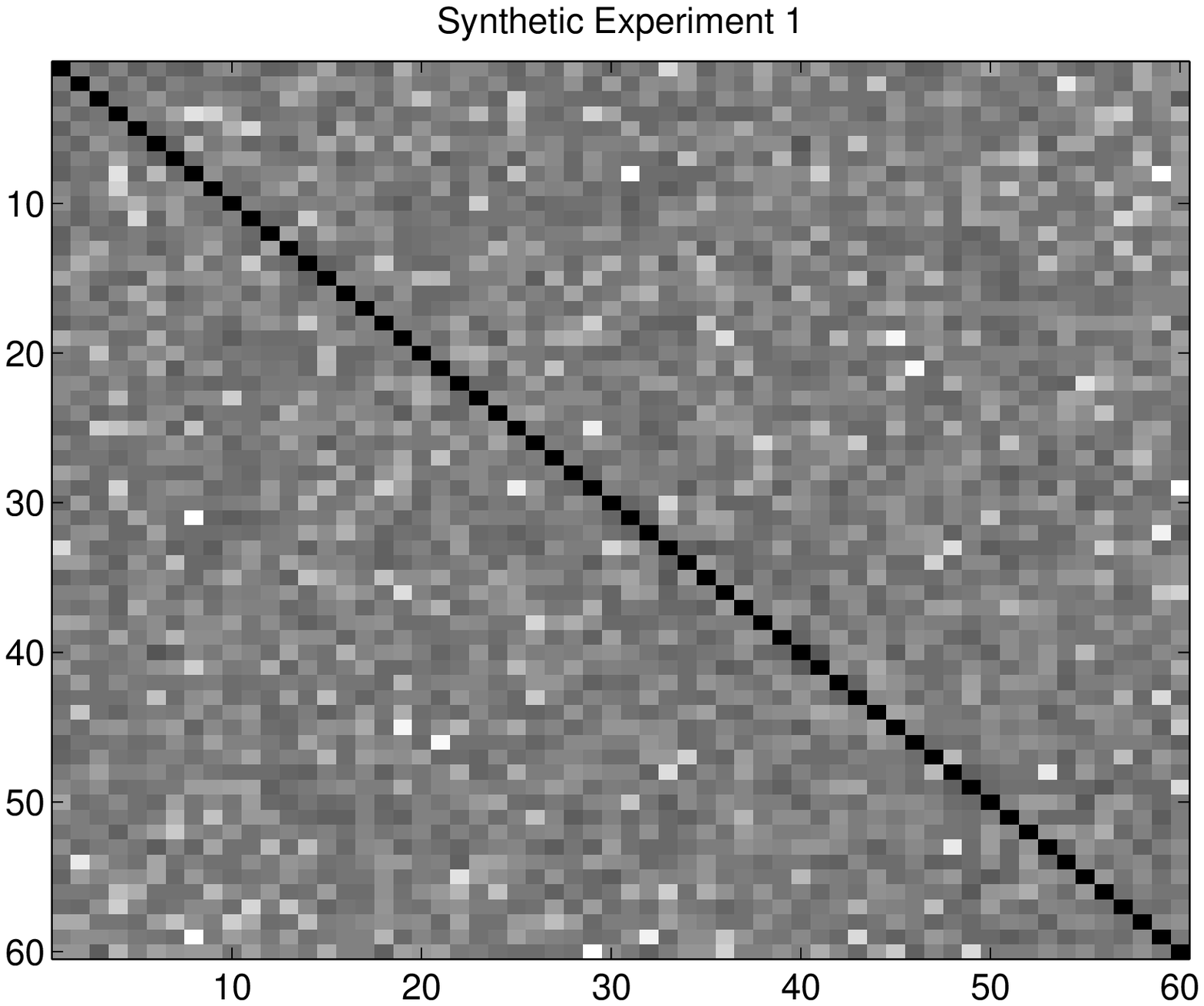} & \includegraphics[width=0.27\columnwidth]{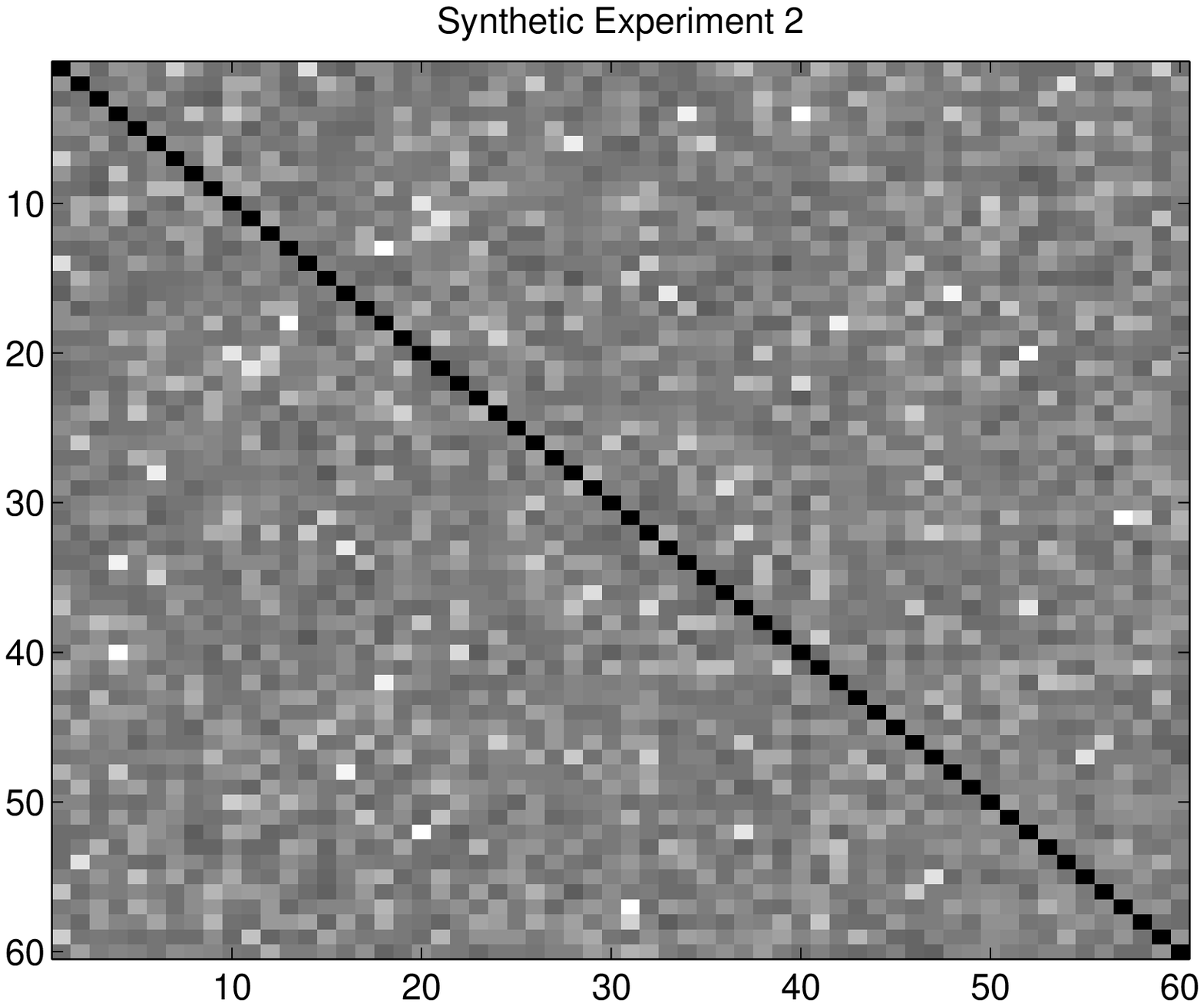} &
\includegraphics[width=0.27\columnwidth]{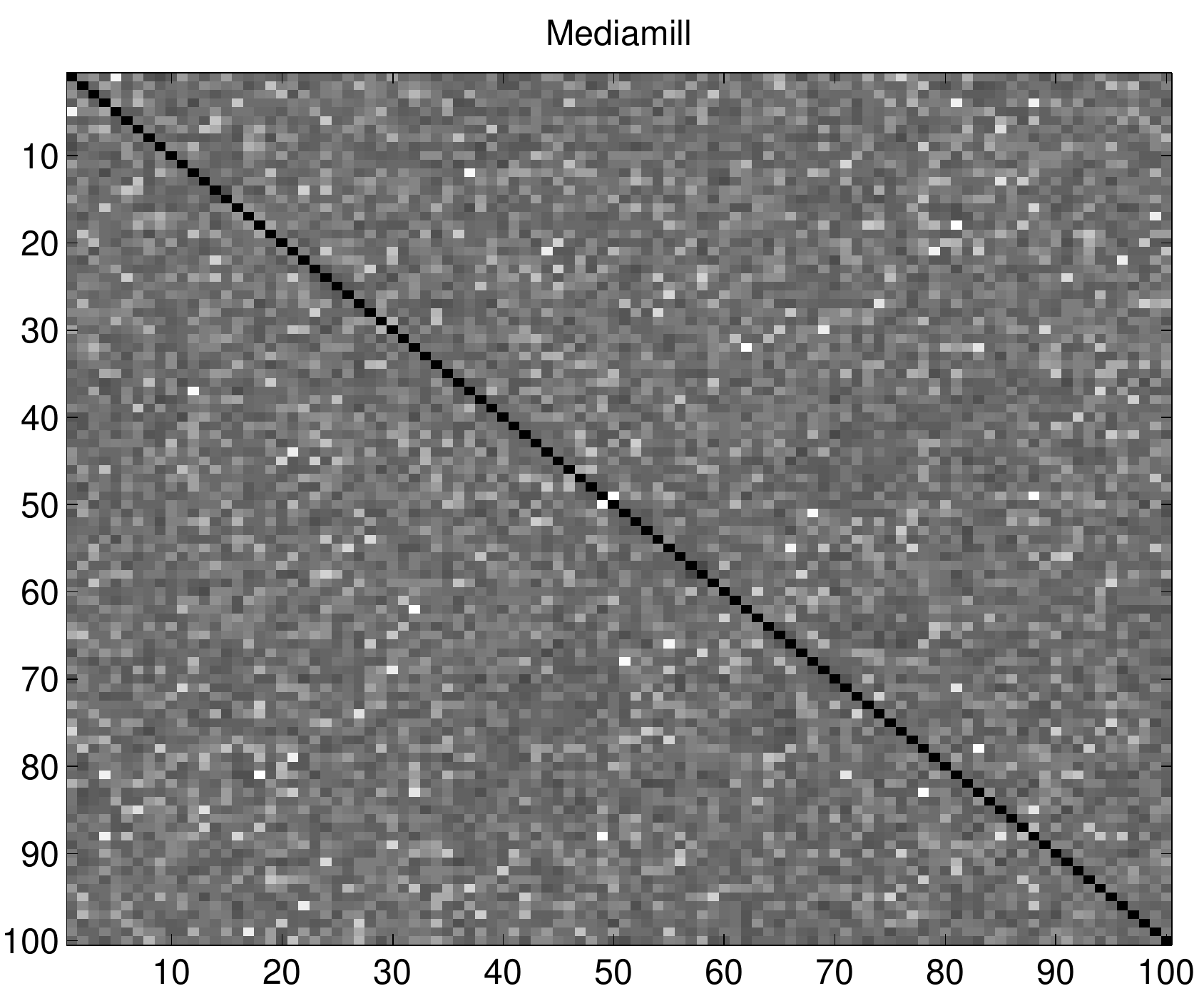}\tabularnewline
% for arxiv upload
%\includegraphics[width=0.27\columnwidth]{example1_aw} & \includegraphics[width=0.27\columnwidth]{example2_aw} &
%\includegraphics[width=0.27\columnwidth]{mediamill_aw}\tabularnewline
\textbf{(a)} & \textbf{(b)} & \textbf{(c)} \tabularnewline
\end{tabular}
\par\end{centering}
\caption{\label{fig:cc-bp}Visualization of the absolute value of the enteries in $\matP\transp \matB\transp \matB \matP$ in one of the runs. Color vary between white and black, with black is $1$ and white is $10^{-5}$.}
\end{figure*}

\subsection{ Synthetic Experiment 1} In this experiment we first draw five random matrices: three matrices $\matG, \matW, \matZ \in \R^{m \times n}$ with independent entries from the normal distribution, and two matrices $\matX, \matY \in \R^{n \times n}$ with independent entries from the uniform distribution on $[0, 1]$. We now set $\matA = \matG \matX + 0.1 \cdot \matW$ and $\matB = \matG \matY + 0.1 \cdot \matZ$. We use the sizes $m=120,000$ and $n=60$. Conceptually,  we first take a random basis (the columns of $\matG$), and linearly transform it in two different ways (by multiplying by $\matX$ and $\matY$). The transformation does not change the space spanned by the bases. We now add to each base some random noise ($0.1 \cdot \matW$ and $0.1 \cdot \matZ$). Since both $\matA$ and $\matB$ essentially span the same column space, only polluted by different noise, we expect $(\matA, \matB)$ to have mostly large canonical correlations (close to $1$), but also a few small ones. Indeed, Figure~\ref{fig:cc-values}(a), which plots the canonical correlations of this pair of matrices, confirms our hypothesis.

Figure~\ref{fig:cc-approx}(a) shows the (signed) error in approximating the canonical correlations, in five different runs. The actual error is always an order of magnitude smaller than the input $\epsilon$; the maximum absolute error is only $0.011$. For large canonical correlations the error is much smaller, and the approximated value is very accurate. For smaller correlations, the error starts to get larger, but it is still an order of magnitude smaller than the actual value for the smallest correlation.

Next, we checked whether $\matA \matW$ and $\matB \matP$ are close to having orthogonal columns, where $\matW$ and $\matP$ contain the canonical weights returned by the proposed algorithm. Figure~\ref{fig:cc-aw}(a) visualizes the entries of $\matW\transp \matA\transp \matA \matW$ and figure~\ref{fig:cc-bp}(a) visualizes the entries of $\matP\transp \matB\transp \matB \matP$ in one of the runs. We see that the diagonal is dominant, and close to 1, and the off diagonal entries are small (but not tiny). The maximum condition number of $\matA \matW$ and $\matB \matP$ we got in the five different runs was $1.08$, indicating the columns are indeed close to be orthogonal.

As for the running time, the proposed algorithm takes about 55\% less time than Bj{\"o}rck and Golub's algorithm ($0.915$ seconds vs. $2.04$ seconds).

\subsection{Synthetic Experiment 2} In this experiment we first draw three random matrices. The first matrix $\matX \in \R^{m \times n}$ has independent entries from the normal distribution. The second matrix $\matY \in \R^{m \times k}$ has independent entries which take value $\pm 1$ with equal probability. The third matrix $\matZ \in \R^{k \times n}$ has independent entries from the uniform distribution on $[0, 1]$. We now set $\matA = \matX + 0.1 \cdot \matY \cdot \left( \one_{k \times n} + \matZ \right)$ and $\matB = \matY$, where $\one_{k \times n}$ is the $k \times n$ all-ones matrix. We use the sizes $m=80,000$, $n=80$ and $k=60$. Here we basically have noise ($\matB$) and a matrix polluted with that noise ($\matA$). So there is some correlation, but really the two subspaces are different; there is one large correlation (almost $1$) and all the rest are small (Figure~\ref{fig:cc-values}(b)).
%When you compute correlations (Figure~\ref{fig:cc-values}(b)) you find one large (almost $1$) and all the rest are small or tiny.

Figure~\ref{fig:cc-approx}(b) shows the (signed) error in approximating the correlations, in five different runs. The actual error is an order of magnitude smaller than the target $\epsilon$; the maximum absolute error is only $0.02$. Again, for the largest canonical correlation (which is close to $1$) the result is very accurate, with tiny errors. For the other correlations it is larger. For tiny correlations the error is about of the same magnitude as the actual value. Interestingly, we observe a bias towards over-estimating the correlations. 

Next, we checked whether $\matA \matW$ and $\matB \matP$ are close to having orthogonal columns, where $\matW$ and $\matP$ contain the canonical weights returned by the proposed algorithm. Figure~\ref{fig:cc-aw}(b) visualizes the entries of $\matW\transp \matA\transp \matA \matW$ and figure~\ref{fig:cc-bp}(b) visualizes the entries of $\matP\transp \matB\transp \matB \matP$ in one of the runs. We see that the diagonal is dominant, and close to 1, and the off diagonal entries are small (but not tiny). The maximum condition number of $\matA \matW$ and $\matB \matP$ we got in the five different runs was $1.08$, indicating the columns are indeed close to be orthogonal.

As for the running time, the proposed algorithm takes about 40\% less time than Bj{\"o}rck and Golub's algorithm ($1.77$ seconds vs. $1.77$ seconds).

\subsection{ Real-life dataset: Mediamill}
We also tested the proposed algorithm on the annotated video dataset from the Mediamill Challenge~\cite{SnoekEtAl06}\footnote{The dataset is publicly available at \url{http://www.csie.ntu.edu.tw/~cjlin/libsvmtools/datasets/multilabel.html\#\# mediamill.}}. Combining the training set and the challenge set, 43907 images are provided, each image is a representative keyframe image of a video shot. The dataset provides 120 features for each image, and the set is annotated with 101 labels. The label matrix is rank-deficient with rank 100. Figure~\ref{fig:cc-values}(c) shows the exact canonical correlations. We see there is a few high correlations, with very strong decay afterwards.

Figure~\ref{fig:cc-approx}(c) shows the (signed) error in approximating the correlations, in five different runs. The maximum absolute error is rather small (only 0.055). For the large correlations, which are the more interesting ones in this context, the error is much smaller, so we have a relatively high accuracy approximation. Again, there is an interesting bias towards over-estimating the correlations. 

Next, we checked whether $\matA \matW$ and $\matB \matP$ are close to having orthogonal columns, where $\matW$ and $\matP$ contain the canonical weights returned by the proposed algorithm. Figure~\ref{fig:cc-aw}(c) visualizes the entries of $\matW\transp \matA\transp \matA \matW$ and figure~\ref{fig:cc-bp}(c) visualizes the entries of $\matP\transp \matB\transp \matB \matP$ in one of the runs. We see that the diagonal is dominant, and close to 1, and the off diagonal  entries are small (but not tiny). The maximum condition number of $\matA \matW$ and $\matB \matP$ we got in the five different runs was $1.23$, which is larger than the previous two examples, but still indicating the columns are not too far from being orthogonal.

As for the running time, the proposed algorithm is considerably faster than Bj{\"o}rck and Golub's algorithm ($0.69$ sec vs. $2.03$ sec).

\subsection{Summary}
The experiments are not exhaustive, but they do suggest the following. First, it appears that the sampling size bounds are rather loose. The algorithm achieves much better approximation errors. Second, there seems to be a connection between the canonical correlation value and the error: for larger correlations the error is smaller. Our bounds fail to capture these phenomena. Finally, the experiments show that the proposed is faster than Bj{\"o}rck and Golub's algorithm {\em in practice} on both synthetic and real-life datasets, even if they are fairly small. We expect the difference to be much larger on big datasets.

\section{Conclusions}
We proved that dimensionality reduction via Randomized Fast Unitary Transforms leads to faster algorithms for Canonical Correlation Analysis, 
beating the seminal SVD-based algorithm of Bj{\"o}rck and Golub. 

The proposed algorithm builds upon a family of similar algorithms which, in recent years, 
led to similar running time improvements for other classical linear algebraic and machine learning problems: 
(i) Least-squares regression~\cite{RT08, BD09,DMMS11,AMT10};
(ii) approximate PCA (via low-rank matrix approximation)~\cite{HMT}; 
(iii) matrix multiplication~\cite{Sar06}; (v) K-means clustering~\cite{BZD10};
(vi) support vector machines~\cite{PBMD13}.

\section*{Acknowledgments}
Haim Avron and Christos Boutsidis acknowledge the support from XDATA program of the Defense Advanced Research Projects Agency (DARPA), administered through Air Force Research Laboratory contract FA8750-12-C-0323. Sivan Toledo was supported by grant 1045/09 from the Israel Science Foundation (founded by the Israel Academy of Sciences and Humanities) and by grant 2010231 from the US-Israel Binational Science Foundation.

\bibliographystyle{plain}

\begin{thebibliography}{10}

\bibitem{AC06}
N.~Ailon and B.~Chazelle.
\newblock Approximate nearest neighbors and the fast johnson-lindenstrauss
  transform.
\newblock In {\em Proceedings of the Symposium on Theory of Computing (STOC)},
  pages 557--563, 2006.

\bibitem{AL08}
N.~Ailon and E.~Liberty.
\newblock Fast dimension reduction using {R}ademacher series on dual {BCH}
  codes.
\newblock In {\em Proceedings of the ACM-SIAM Symposium on Discrete Algorithms
  (SODA)}, 2008.

\bibitem{AMT10}
H.~Avron, P.~Maymounkov, and S.~Toledo.
\newblock {Blendenpik: Supercharging {LAPACK}'s least-squares solver}.
\newblock {\em SIAM Journal on Scientific Computing}, 32(3):1217--1236, 2010.

\bibitem{DISJ}
Z.~Bar-Yossef, T.~S. Jayram, R.~Kumar, and D.~Sivakumar.
\newblock An information statistics approach to data stream and communication
  complexity.
\newblock {\em J. Comput. Syst. Sci.}, 68(4):702--732, 2004.

\bibitem{BG73}
A.~Bj{\"o}rck and G.H. Golub.
\newblock Numerical methods for computing angles between linear subspaces.
\newblock {\em Mathematics of Computation}, 27(123):579--594, 1973.

\bibitem{BD09}
C.~Boutsidis and P.~Drineas.
\newblock {Random projections for the nonnegative least-squares problem}.
\newblock {\em Linear Algebra and its Applications}, 431(5-7):760--771, 2009.

\bibitem{BZD10}
C.~Boutsidis, A.~Zouzias, and P.~Drineas.
\newblock Random projections for $k$-means clustering.
\newblock In {\em Neural Information Processing Systems (NIPS)}, 2010.

\bibitem{CC:story}
A.~Chattopadhyay and T.~Pitassi.
\newblock {The story of set disjointness}.
\newblock {\em SIGACT News}, 41(3):59--85, 2010.

\bibitem{CKLS09}
K.~Chaudhuri, S.~M. Kakade, K.~Livescu, and K.~Sridharan.
\newblock Multi-view clustering via canonical correlation analysis.
\newblock In {\em International Conference in Machine Learning (ICML)}, pages
  129--136, 2009.

\bibitem{CW12}
K.~L. Clarkson and D.~P. Woodruff.
\newblock {Low Rank Approximation and Regression in Input Sparsity Time}.
\newblock In {\em Proceedings of the Symposium on Theory of Computing (STOC)},
  2013.

\bibitem{DRFU12}
P.~Dhillon, J.~Rodu, D.~Foster, and L.~Ungar.
\newblock Two step {CCA}: A new spectral method for estimating vector models of
  words.
\newblock In {\em Proceedings of the 29th International Conference on Machine
  Learning}, ICML'12, 2012.

\bibitem{DFU11}
P.~S. Dhillon, D.~Foster, and L.~Ungar.
\newblock Multi-view learning of word embeddings via {CCA}.
\newblock In {\em Neural Information Processing Systems (NIPS)}, 2011.

\bibitem{DMMW12}
P.~Drineas, M.~Magdon-Ismail, M.~W. Mahoney, and D.~P. Woodruff.
\newblock Fast approximation of matrix coherence and statistical leverage.
\newblock In {\em International Conference in Machine Learning (ICML)}, 2012.

\bibitem{DMMS11}
P.~Drineas, M.W. Mahoney, S.~Muthukrishnan, and T.~Sarl{\'o}s.
\newblock Faster least squares approximation.
\newblock {\em Numerische Mathematik}, 117(2):217--249, 2011.

\bibitem{EI95}
S.~Eisenstat and I.~Ipsen.
\newblock Relative perturbation techniques for singular value problems.
\newblock {\em SIAM Journal on Numerical Analysis}, 32:1972--1988, 1995.

\bibitem{GZ95}
G.H. Golub and H.~Zha.
\newblock The canonical correlations of matrix pairs and their numerical
  computation.
\newblock {\em IMA Volumes in Mathematics and its Applications}, 69:27--27,
  1995.

\bibitem{HMT}
N.~Halko, P.G. Martinsson, and J.A. Tropp.
\newblock Finding structure with randomness: Probabilistic algorithms for
  constructing approximate matrix decompositions.
\newblock {\em SIAM Review}, 53(2):217--288, 2011.

\bibitem{HJ85}
R.~A. Horn and C.~R. Johnson.
\newblock {\em Matrix Analysis}.
\newblock Cambridge University Press, 1985.

\bibitem{Hot36}
H.~Hotelling.
\newblock Relations between two sets of variates.
\newblock {\em Biometrika}, 28(3/4):321--377, 1936.

\bibitem{IW12}
I.~Ipsen and T.. Wentworth.
\newblock The effect of coherence on sampling from matrices with orthonormal
  columns, and preconditioned least squares problems.
\newblock {\em Arxiv preprint arXiv:1203.4809}, 2012.

\bibitem{KKC07}
T.-K. Kim, J.~Kittler, and R.~Cipolla.
\newblock Discriminative learning and recognition of image set classes using
  canonical correlations.
\newblock {\em IEEE Trans. Pattern Anal. Mach. Intell.}, 29(6):1005--1018,
  2007.

\bibitem{book:CC}
E.~Kushilevitz and N.~Nisan.
\newblock {\em {Communication complexity}}.
\newblock Cambridge University Press, New York, NY, USA, 1997.

\bibitem{MM13}
X.~Meng and M.~W. Mahoney.
\newblock {Low-distortion Subspace Embeddings in Input-sparsity Time and
  Applications to Robust Linear Regression}.
\newblock In {\em Proceedings of the Symposium on Theory of Computing (STOC)},
  2013.

\bibitem{PBMD13}
S.~Paul, C.~Boutsidis, M.~Magdon-Ismail, and P.~Drineas.
\newblock Random projections for support vector machines.
\newblock In {\em International Conference on Artificial Intelligence and
  Statistics (AISTATS)}, 2013.

\bibitem{RT08}
V.~Rokhlin and M.~Tygert.
\newblock {A fast randomized algorithm for overdetermined linear least-squares
  regression}.
\newblock {\em Proceedings of the National Academy of Sciences}, 105(36):13212,
  2008.

\bibitem{Sar06}
T.~Sarl\'{o}s.
\newblock Improved approximation algorithms for large matrices via random
  projections.
\newblock In {\em Proceedings of the Symposium on Foundations of Computer
  Science (FOCS)}, 2006.

\bibitem{SnoekEtAl06}
C.~G.~M. Snoek, M.~Worring, J.~C. van Gemert, J.~M. Geusebroek, and A.~W.~M.
  Smeulders.
\newblock The challenge problem for automated detection of 101 semantic
  concepts in multimedia.
\newblock In {\em Proceedings of the ACM International Conference on
  Multimedia}, pages 421--430, 2006.

\bibitem{SFGT12}
Y.~Su, Y.~Fu, X.~Gao, and Q.~Tian.
\newblock Discriminant learning through multiple principal angles for visual
  recognition.
\newblock {\em IEEE Transactions on Image Processing}, 21(3):1381 --1390, March
  2012.

\bibitem{SCY10}
L.~Sun, B.~Ceran, and J.~Ye.
\newblock {A scalable two-stage approach for a class of dimensionality
  reduction techniques}.
\newblock In {\em ACM SIGKDD Conference on Knowledge Discovery and Data Mining
  (KDD)}, pages 313--322, 2010.

\bibitem{SJY08}
L.~Sun, S.~Ji, and J.~Ye.
\newblock A least squares formulation for canonical correlation analysis.
\newblock In {\em International Conference in Machine Learning (ICML)}, pages
  1024--1031, 2008.

\bibitem{TR10}
A.~Talwalkar and A.~Rostamizadeh.
\newblock Matrix coherence and the {N}ystr{\"o}m method.
\newblock In {\em UAI}, pages 572--579, 2010.

\bibitem{Tro11}
J.~A. Tropp.
\newblock Improved analysis of the subsampled randomized {H}adamard transform.
\newblock {\em Adv. Adapt. Data Anal., special issue, ``Sparse Representation
  of Data and Images''}, 2011.

\bibitem{CC:Yao}
A.~C.-C. Yao.
\newblock {Some complexity questions related to distributive computing
  (Preliminary Report)}.
\newblock In {\em Proceedings of the Symposium on Theory of Computing (STOC)},
  pages 209--213, 1979.

\end{thebibliography}

\end{document}